\documentclass[letterpaper, 11pt, leqno]{article}

\usepackage{amssymb}
\usepackage{amsmath}
\usepackage{amsfonts}
\usepackage{amsthm}
\usepackage{mathtools}
\usepackage{graphicx}
\usepackage{enumerate}
\usepackage[margin=1.25in]{geometry}
\usepackage{setspace}
\usepackage[applemac]{inputenc}
\usepackage{dsfont}
\usepackage{natbib}
\usepackage{mathrsfs}
\usepackage{thmtools}
\usepackage{color}
\usepackage{comment}
\usepackage{latexsym}
\usepackage{cancel}
\usepackage{adjustbox}

\usepackage{arydshln}
\usepackage{booktabs}

\usepackage{geometry} 
\usepackage{xcolor, graphicx, colortbl} 
\usepackage{nicematrix}
\usepackage{tikz}
\usetikzlibrary{patterns}
\usetikzlibrary{trees}

\definecolor{purple}{RGB}{85, 6,139}
\definecolor{teal}{RGB}{2,108,128}
\definecolor{lavender}{RGB}{129, 102, 122}
\definecolor{carolina blue}{RGB}{68, 157, 209}
\definecolor{phthalo blue}{RGB}{2, 8, 135}
\definecolor{purple2}{RGB}{149, 96, 219}
\definecolor{green1}{RGB}{96, 219, 117}
\definecolor{orange1}{RGB}{208,70,0}

\DeclareMathOperator*{\argmax}{arg\,max}
\DeclareMathOperator*{\argmin}{arg\,min}

\newcommand{\R}{\mathbb{R}}

\newcommand{\F}{\mathcal{F}}

\newtheorem{thm}{Theorem}
\newtheorem{prps}{Proposition}
\newtheorem{lem}{Lemma}

\theoremstyle{definition}
\newtheorem{axm}{\sc Axiom}
\newtheorem{defn}{Definition}
\newtheorem{cor}{Corollary}

\newtheorem{exm}{Example}

\usepackage{nameref}
\usepackage[pagebackref]{hyperref}
\hypersetup{
breaklinks = true,
  colorlinks   = true, 
  urlcolor     = purple, 
  linkcolor    = teal, 
  citecolor   = gray 
}

\title{Inertial Updating\thanks{Dominiak: Virginia Tech (dominiak@vt.edu); Kovach: Virginia Tech (mkovach@vt.edu); Tserenjigmid: UC Santa Cruz (gtserenj@ucsc.edu). This paper subsumes ``Ordered Surprises and Conditional Probability Systems" and Section 6 of ``Minimum Distance Belief Updating with General Information" by the same authors. We are very grateful to David Freeman, Paolo Ghirardato, Faruk Gul, Edi Karni, Shaowei Ke, Yusufcan Masatlioglu, Pietro Ortoleva, Burkhard Schipper, Dong Wei, and Chen Zhao for valuable comments and discussions, as well as the seminar participants at UC Riverside, Texas A\&M, ISI Delhi, University of Michigan, and Purdue University.}}
\vspace{5 mm}
\author{\centering Adam Dominiak  \and Matthew Kovach \and Gerelt Tserenjigmid}
\date{\centering\today} 

\begin{document}

\maketitle



\noindent{\textbf{Abstract:} {\onehalfspacing We introduce and characterize \emph{inertial updating} of beliefs. Under inertial updating, a decision maker (DM) chooses a belief that minimizes the subjective distance between their prior belief and the set of beliefs consistent with the observed event. Importantly, by varying the subjective notion of distance, inertial updating provides a unifying framework that nests three different types of belief updating: (i) Bayesian updating, (ii) non-Bayesian updating rules, and (iii) updating rules for events with zero probability, including the conditional probability system (CPS) of \cite{Myerson1886a, Myerson1886b}. We demonstrate that our model is behaviorally equivalent to the Hypothesis Testing model (HT) of \cite{ortoleva2012}, clarifying the connection between HT and CPS and non-Bayesian updating models.  We apply our model to a persuasion game. 

\vspace{2 mm}

\vspace{2 mm}


\newpage

\section{Introduction}
\onehalfspacing

How decision makers revise their beliefs after receiving information is a foundational problem in economics and game theory. While the benchmark model of Bayesian updating is broadly appealing for a variety of reasons, it has two major issues.  First, it is descriptively limited; there is robust experimental evidence that people's beliefs systematically deviate from what Bayesian updating prescribes.\footnote{For experimental evidence, see \cite{KahnemanTversky1972}, \cite{KT1983}, \cite{Camerer1987}, \cite{Eil2011}, along with surveys by \cite{Camerer1995} and \cite{benjamin2019}.} 
Second, it is \emph{incomplete}; a well-known limitation of Bayesian updating is that it is not defined for zero-probability events.\footnote{This is an especially important issue in dynamic games of incomplete information, as particular off-path beliefs are used to support certain equilibria. Accordingly, complete theories of belief updating, such as the Conditional Probability System introduced by \cite{Myerson1886a, Myerson1886b}, have been proposed.} We resolve these limitations of Bayesian updating by introducing the Inertial Updating (\nameref{IEU}) representation: a \emph{complete} theory of belief updating that unifies Bayesian and non-Bayesian updating rules.


 \nameref{IEU} addresses both of these issues by recasting belief updating as an optimization problem; belief updating is transformed into a problem of belief selection satisfying two intuitive properties. First, our DM must select a belief that is consistent with the information, hence information induces a constraint set. Second, our DM selects a belief that is closest to her current belief according to a subjective distance function.\footnote{For ease of exposition, we use the term ``distance function," which may not satisfy the triangle inequality in our case.} Slightly more formally, given a prior $\mu$ over a set of states $S$ and any event $E \subset S$, her new belief $\mu_E$ is the distribution over $E$ that is ``closest'' to $\mu$ among all of the probability distributions over $E$. Since our DM minimizes the change in her beliefs relative to her prior, we refer to this behavior as Inertial Updating.  Since our DM utilizes a subjective notion of distance, our framework is flexible enough to encompass a variety of updating patterns. We provide a complete behavioral analysis of \nameref{IEU} and demonstrate that it provides a unifying framework to capture various belief updating rules in the literature. 

The \nameref{IEU} representation is characterized by three axioms (see \autoref{axioms}).  The first two postulates are standard: \nameref{SEU} imposes a subjective expected utility representation for each conditional preference $\succsim_E$, and \nameref{CON} ensures that for any event $E$, the DM only considers states within $E$ possible (i.e., $\mu_E\in \Delta(E)$).  The third axiom, \nameref{DCoh}, was introduced by \cite{ortoleva2012} to characterize the Hypothesis Testing model (HT).\footnote{In the HT, an agent's behavior is consistent with SEU, yet she also has a second-order belief and thus has multiple beliefs in mind. She updates her prior according to Bayes' rule if she receives ``expected'' information. When information is ``unexpected,'' she rejects her prior and uses her second-order belief to select a new belief according to a maximum likelihood rule.  Thus an HT agent is essentially Bayesian, but is nevertheless open to fundamentally shifting her worldview.} To interpret this axiom, say that an event $A$ is \emph{revealed implied by} event $B$ if every state that the DM believes is possible after learning $B$ is also an element of $A$.  That is, once the DM learns that the ``true state is contained in $B$,'' she is also convinced that that ``true state is contained in $A$,'' and therefore $A^c$ is believed to be null after $B$.  \nameref{DCoh} requires that this revealed implication over events is acyclic. 

Our main result, \autoref{repthrm}, shows that the preceding three axioms are necessary and sufficient for the \nameref{IEU} representation. Our proof is based on an extension of Afriat's theorem (\cite{afriat1967construction}, \cite{varian1982nonparametric}) for general budget sets due to \cite{matzkin1991axioms}. We are able to apply this theorem by showing that \nameref{DCoh} implies that the data set of ``belief choices'' satisfies the Strong Axiom of Revealed Preferences (SARP).  As in Afriat's theorem, we get continuity and strict convexity of the distance function for free. A corollary of our theorem is that \nameref{IEU} and HT are behaviorally equivalent, despite their stark difference in appearance and the significantly different proof techniques.

One key feature of \nameref{IEU} is that it is descriptively rich; \nameref{IEU} accommodates Bayesian and non-Bayesian updating.  While it is well known that \nameref{DC} ensures Bayesian updating, we provide a complimentary result showing that distance functions that are generalizations of the celebrated Kullback-Leibler (KL) divergence deliver posteriors that are consistent with Bayesian updating.  We then build upon this insight to define a family of non-Bayesian updating rules that we call \nameref{hBayes}. 


A variety of updating biases fall under \nameref{hBayes}. In particular, \nameref{hBayes} updating has a non-trivial connection to the well-known $\alpha-\beta$ rule from \cite{grether1980}, capturing forms of under- or over-reaction. The \nameref{hBayes} can also allow for asymmetric reactions, along with features of confirmation bias. Further, this rule allows for history-dependent updating, and therefore it can capture a wide array of context effects. We provide a behavioral characterization of \nameref{hBayes} via two axioms, both of which are weaker than  \nameref{DC}. The characterization of \nameref{hBayes} and a discussion of the preceding examples can be found in section \ref{sec:hBayes}.

The other key feature of \nameref{IEU} is that it is a complete theory of updating: conditional beliefs are well-defined for all events. This follows because the DM's notion of distance is well-defined for all distributions. Of course, we are not the first to propose a complete theory of updating. The most prominent complete theory is Myerson's Conditional Probability System (CPS) \citep{Myerson1886a, Myerson1886b}, which was motivated by the Sequential Equilibria of \cite{KrepsWilson1982}.

We provide a simple behavioral foundation for CPS in section \ref{CPSsection}. Our characterization relies upon a novel axiom, \nameref{CC}, that implies \nameref{DC} among the non-null events and extends this consistency to ``conditionally non-null'' events. We then show that CPS is a special case of \nameref{IEU} by providing an explicit distance function that generates any CPS. Because \nameref{IEU} and HT are behaviorally equivalent, this also establishes that the CPS is a special case of HT.

The relations between HT and other models of updating such as CPS and Grether's $\alpha-\beta$ rule were not known previously, partly due to stark differences in their representations. By recasting the problem of updating as an optimization problem, our model and results clarify the exact relations between HT, CPS, Grether's $\alpha-\beta$ rule, and Distorted Bayesian in general.






We apply \nameref{IEU} to settings with signal structures and provide a distance function that generalizes the $\alpha-\beta$ rule from \cite{grether1980} in \autoref{Signal}. The generalization of Grether's rule uses two distortion functions, a \emph{prior distortion} $g$ and \emph{signal distortion} $f$, and reduces to Grether's rule when both distortions are power functions. We discuss how over-reaction and under-reaction to news can be captured simultaneously. 

We use this distorted Bayesian distance to analyze the effect of non-Bayesian belief updating on the optimal signal structure in the Bayesian persuasion games of \cite{kamenica2011bayesian} (section \ref{sec-app}). We find that the way it distorts prior probabilities, $g$, has no qualitative impact on the optimal signal structure, whereas the optimal signal structure depends critically on the curvature of the signal distortion $f$. In particular, the set of states at which the sender is fully revealing when $f$ is concave is drastically different from when $f$ is strictly convex. 

We close the paper by introducing a generalization of \nameref{IEU} that relaxes Consequentialism (\autoref{sect-wiu}) and discussing related literature (\autoref{literature}). 



\section{Model}\label{model}
\subsection{Basic Setup}
We study choice under uncertainty in the framework of \cite{anscombeaumann1963}. A DM faces uncertainty described by a nonempty and finite set of states of nature $S=\{s_1,\ldots,s_n\}$.\footnote{We focus on a finite state space as it is more standard for decision theoretic analysis and general enough for most economic applications, but we can easily extend our model to an infinite state space.}  Let $\Sigma$ be an arbitrary collection of nonempty subsets of $S$ such that $S\in\Sigma$. Any element $E$ of $\Sigma$ is called an event. Let $X$ be a nonempty, finite set of outcomes and $\Delta(X)$ be the set of all lotteries over $X$, i.e., $\Delta(X) := \big\{ p : X \rightarrow [0,1] \mid \sum_{x \in X}p(x)=1 \big\}$.


We are interested in a DM's preference over acts, which are mappings $f: S \to \Delta(X)$ that assigns a lottery to each state. The set of all acts is $\F:=\{f: S\to \Delta(X)\}$. Any act $f$ that assigns the same lottery to all states ($f(s)=p$ for all $s \in S$) is called a constant act. Using a standard abuse of notation, we denote by $p \in \F$ the corresponding constant act. Hence, we can identify the set of lotteries $\Delta(X)$ with the constant acts.  We define mixed lotteries and acts in the usual way: for any $\lambda \in [0,1]$, $\lambda p+ (1-\lambda) q$ is the lottery providing $x$ with probability  $\lambda p(x) + (1-\lambda) q(x)$, and  $\lambda f + (1-\lambda)g$ is the act that yields $\lambda f(s) + (1-\lambda)g(s)$ in state $s$. Moreover, for any $E \in \Sigma$, and $f, h \in \F$,  $f E h$ denotes that conditional act that returns $f(s)$ for $s \in E$ and $h(s)$ otherwise.

The DM's behavior is depicted by a family of preference relations $\{\succsim_E\}_{E\in\Sigma}$, each defined over $\F$.  We write $\succsim$ in place of $\succsim_S$, and we call $\succsim$ the initial preference. As usual, for each $E \in \Sigma$,  $\succ_E$ and $\sim_E$ are the asymmetric and symmetric parts of $\succsim_E$, respectively.  We say that $E$ is $\succsim$-null (or simply null) if $fEg\sim g$ for any $f, g\in \F$. Otherwise, $E$ is non-null. Similarly, we say $E$ is $\succsim_A$-null if $fEg\sim_A g$ for any $f, g\in \F$. If $E$ is not $\succsim_A$-null, then it is $\succsim_A$-non-null. 

We denote by $\Delta(S)$ the set of all probability distributions on $S$. For notational convenience, for each $\mu\in \Delta(S)$ and each $s_i \in S$, we will sometimes write $\mu_i$ in place of $\mu(s_i)$: the probability of state $s_i$ according to $\mu$.  For any $ \pi \in \Delta(S)$, let $\text{sp}(\pi)$ denote the support of  $\pi$. For any $\mu$ and event $E$ such that $\mu(E)>0$, let $\text{BU}(\mu, E)$ denote the Bayesian update of $\mu$ conditional on $E$.

Finally, let $\|\cdot\|$ denote the Euclidean norm. For any set $A$ and a function $d$ on $A$, we write  $\arg\min d(A)=\{x\in A \mid d(y)\ge d(x)\text{ for any }y\in A\}$ (whenever this is well-defined).

\subsection{Inertial Updating}

When the DM observes an event  $E \in \Sigma$, she revises her initial preference $\succsim$ to a conditional preference denoted $\succsim_E$. This setting is quite general and incorporates the standard signal structure as a special case.\footnote{In particular, when $S=\Omega \times M$, for a set of payoff relevant states $\Omega$ and signals $M$, the signal $m$ corresponds to the event $\{(\omega,m) \in S \mid 
\omega \in \Omega \}$.} We provide additional analysis of this special case in \autoref{Signal}. 

Rather than specify a specific formula that generates the DM's conditional beliefs (e.g., Bayes' rule, Grether's $\alpha-\beta$ rule), \nameref{IEU} imposes general restrictions on the revision process. That is, \nameref{IEU} requires that her new belief is (i) consistent with the information and (ii) of minimal distance to her prior, while allowing the distance notion to be subjective.  We now formally define our notion of distance. 

\begin{defn}[Distance Function] A function $d:\Delta(S) \to\mathds{R}$ is a \textbf{distance function} with respect to $\mu\in \Delta(S)$, denoted by $d_\mu$, if $d_{\mu}(\mu)<d_{\mu}(\pi)$ for any $\pi\in\Delta\setminus\{\mu\}$. 
\end{defn}
 
The only condition required of the the distance function is that the prior is the global minimizer among all beliefs. This is a simple coherence property, because otherwise a DM should immediately adopt some other belief. Equipped with this notion of distance, we now introduce the \nameref{IEU} representation. For ease of exposition, we use the term ``distance function" even though $d$ may not satisfy the triangle inequality.

\begin{defn}[{\textbf{IU}}]\label{IEU} A family of preference relations $\{\succsim_E\}_{E\in\Sigma}$ admits an \textbf{Inertial Updating} representation if there are a Bernoulli utility function $u:X\to\mathds{R}$, a prior $\mu \in \Delta(S)$, a distance function $d_\mu:\Delta(S)\to\mathds{R}$ such that for each $E \in \Sigma$, the preference relation $\succsim_E$ admits a SEU representation with $(u, \mu_E)$, meaning that for any $f, g\in \F$, \begin{equation}
f \succsim_E g \quad \text{if and only if} \quad \sum_{s\in E}\mu_E(s) u\big(f(s)\big)\ge \sum_{s\in E}\mu_E(s) u\big(g(s)\big), 
\end{equation}where
\begin{equation}\mu_E\equiv \argmin_{\pi \in \Delta(E)}  d_\mu(\pi).\end{equation}
\end{defn}

Since the prior is the global minimizer of $d_{\mu}$, $\mu = \argmin_{\pi \in \Delta(S)}  d_\mu(\pi)$.  For any $E \in \Sigma$, the constraint $\Delta(E)$ is convex, and so $\argmin_{\pi \in \Delta(E)}  d_\mu(\pi)$ will be unique whenever $d_\mu$ is strictly quasi-convex. In fact, the following much weaker condition will suffice: for any $\pi, \pi'\in\Delta(S)$ with $\pi\neq \pi'$, if $d_\mu(\pi)=d_\mu(\pi')$, then there is $\alpha\in (0, 1)$ such that $d_\mu(\alpha\pi+(1-\alpha)\pi')<d_\mu(\pi)$. As our main theorem shows, we get continuity and strict convexity of $d$ for free. Hence, we will not impose any additional properties on $d$.\footnote{The distance functions in Definitions 3-4 are convex, and the distance functions in Definitions 5-7 are strictly convex.}

\subsection{Notions of Distance}\label{notdis}

By allowing for a subjective notion of distance, the \nameref{IEU} generalizes Bayesian updating while also providing a unifying approach to non-Bayesian updating rules.  In this section, we discuss a few examples of distance functions and the beliefs they generate. We begin by introducing a Bayesian distance, which will also be useful in defining non-Bayesian distances later. 

\begin{defn}[\textbf{Bayesian Divergence}]\label{GBex} For any strictly increasing and strictly concave function $\sigma:\mathds{R}_{+}\to\mathds{R}$, let $d_{\mu}$ be given by
\begin{equation}\label{BDeq}
d_\mu(\pi)=-\sum^n_{i=1}\mu_i\, \sigma\left(\frac{\pi_i}{\mu_i}\right).
\end{equation}
\end{defn}

Our first proposition shows that any \nameref{GBex} will generate Bayesian posteriors for all non-null events.\footnote{Bayesian divergence must be modified to be part of an \nameref{IEU} representation; i.e., to yield a complete updating rule. For example, see \autoref{SDBD} for one such way to extend $d_{\mu}$.}

\begin{prps}\label{prp:BD} For any non-null $E\in\Sigma$,  
\[\mu_E=\argmin_{\pi\in\Delta(E)} -\sum^n_{i=1}\mu_i\, \sigma\left(\frac{\pi_i}{\mu_i}\right)=\text{BU}(\mu, E)\]
\end{prps}

Notably, \autoref{BDeq} ``includes" the KL divergence as a special case ($\sigma(x) = \ln(x)$). However, since $\ln(0)=-\infty$, the KL divergence is not well-defined when $\text{sp}(\mu)\subseteq\text{sp}(\pi)$. Therefore, we focus our attention on $\sigma$ that are well defined on $\mathds{R}_{+}$. For example, $\sigma(x)=\ln(\alpha\, x+\beta)$ where $\alpha, \beta>0$ is a well-defined, strictly increasing, and strictly concave function. Alternatively, $\sigma(x)=\frac{x^\alpha-1}{1-\alpha}$ (resulting in the Renyi divergence) is well-defined, strictly increasing, and strictly concave when $\alpha\in (0, 1)$. 

We now introduce the following notation to simplify our exposition. 

\medskip
\noindent\textbf{Notation.} The \textbf{Bayesian function} for a given $\sigma$ is denoted by $\beta^\sigma: \mathbb{R}^n_{+}\times \mathbb{R}^n_{+}\to \mathbb{R}$; i.e., 
\[\beta^{\sigma}(\textbf{x}, \textbf{y})=-\sum^n_{i=1}x_i\, \sigma\left(\frac{y_i}{x_i}\right)\text{ for any }\mathbf{x}, \mathbf{y}\in\mathbb{R}^n_{+}.\] 
The \textbf{Bayesian update} of $\textbf{x}$ on $E$ is denoted by 
\[\text{BU}(\textbf{x}, E)=\left(\frac{x_i\,\mathds{1}\{i\in E\}}{\sum_{j\in E}x_j}\right)_{i\in S}\text{ for any }\mathbf{x}\in \mathbb{R}^n_{+}\text{ with }\sum_{j\in E}x_j>0.\]
Note that $\mathbf{x}$ and $\mathbf{y}$ are not necessarily probability distributions. 

Following the intuition from \nameref{GBex}, we can introduce a distorted version of this distance notion to capture non-Bayesian beliefs. 

%

\begin{defn}[\textbf{Distorted Bayesian}]\label{hBayes} An \nameref{IEU} DM admits a Distorted Bayesian distance if \[d_\mu(\pi)=\beta^\sigma(\delta(\mu), \pi)
\] where $\delta:[0,1]\to\mathbb{R}_{+}$ and $\sigma$ is strictly increasing and strictly concave. Then by Proposition 1, \begin{equation}\mu_E=\text{BU}(\delta(\mu), E)\end{equation} for any non-null $E\in\Sigma$. Further, we say that this distance is \textbf{Monotonic} if $\delta$ is strictly increasing. 
\end{defn}
If $\delta>0$, then we also have $\mu_E=\text{BU}(\delta(\mu), E)$ for any $E$, resulting in a complete theory of belief updating.\footnote{Otherwise, the distance must be modified slightly. See, for example, \autoref{SDBD}.} For example, suppose $\delta$ is defined as follows: $\delta(t)=t+\epsilon\,\mathds{1}\{t=0\}$ where $\epsilon$ is small enough. Then $\mu_E$ is approximately equivalent to $\text{BU}(\mu, E)$ when $E$ is non-null and  when $E$ is a null-event, $\mu_E$ is equivalent to $\text{BU}(\mu^*, E)$  where $\mu^*$ is the uniform distribution over $S$. This example approximates a special case of Myerson's CPS introduced in \autoref{SDBD}.

The \nameref{hBayes} distance notion captures non-Bayesian updating through the distortion function $\delta$.\footnote{For example, $\delta$ captures the DM's imperfect memory or recall of her previously updated belief -- prior (e.g., see \cite{mullainathan2002memory}, \cite{wilson2014bounded}, \cite{gennaioli2010comes}, and  \cite{bordalo2016stereotypes}.} Intuitively, such an agent behaves as if they apply Bayes' rule to a distorted prior. When $\delta(x)=x^{\alpha}$, this corresponds to a special case of Grether's $\alpha - \beta$ rule  \citep{grether1980} where $\alpha=\beta$. For $\alpha<1$, this captures under-reaction to information and base-rate neglect, while $\alpha>1$ captures over-reaction to information. In \autoref{Signal}, we show that our model nests the general version of Grether's $\alpha - \beta$ rule. It is also straightforward to generalize $\delta$ to capture a variety of belief distortions, including asymmetric reactions based on prior beliefs like confirmation bias (\'a la \cite{rabin1999}) or over(under) reaction to small(large) probabilities (\cite{kahneman1979prospect}). 

In section \ref{sec:hBayes} we characterize \nameref{hBayes} and Monotonic \nameref{hBayes}.  Although $\delta$ is independent of the realized event, the \nameref{hBayes} distance can capture features of history or reference dependence. 

\begin{defn}[\textbf{Mixed Bayesian}]\label{mixedbayes} Let $d_{\mu}$ be given by
\begin{equation}
d_\mu(\pi)=\beta^\sigma(\mu+\rho, \pi),
\end{equation}
where $\sigma$ is strictly increasing and strictly concave and $\text{sp}(\mu)\cup\text{sp}(\rho)=S$. Then for any $E\in\Sigma$, by Proposition 1,
\[\mu_E=\text{BU}(\mu+\rho, E)=\alpha(E)\,\text{BU}(\mu, E)+(1-\alpha(E))\,\text{BU}(\rho, E),\]
where $\alpha(E)=\frac{\mu(E)}{\mu(E)+\rho(E)}$. 

\end{defn}

Notice that $\text{sp}(\mu)\cup\text{sp}(\rho)=S$ ensures that \nameref{mixedbayes} yields a complete updating rule; it is defined for all events. When $E$ is a null-event, $\mu_E=\text{BU}(\rho, E)$. Through $\rho$, the Mixed Bayesian distance can capture motivated reasoning \citet{Kunda1990} or wishful thinking (\citet{Mayraz2011experiment, caplin2019wishful, Kovach2020_wishful}).

To illustrate other forms of \nameref{IEU} updating rules for zero-probability events, we can define a support-dependent Bayesian divergence. 


\begin{defn}[\textbf{Support-Dependent Bayesian Divergence}]\label{SDBD} Let
\[d_\mu(\pi)=
\begin{cases}
\beta^\sigma(\mu, \pi)&\text{ if }\mu(\text{sp}(\pi))>0,\\
\beta^\sigma(\mu^*, \pi)+\sigma(1)+|\sigma(0)|&\text{ otherwise},
\end{cases}\]
for $\mu^*$ with $\text{sp}(\mu)\cup \text{sp}(\mu^*)=S$.\end{defn}

\begin{prps}\label{os2} For any $E\in \Sigma$, 
\[\mu_{E}=
\begin{cases}
\text{BU}(\mu, E)&\text{ if }\mu(E)>0,\\
\text{BU}(\mu^*, E)&\text{ otherwise}.\end{cases}
\]
\end{prps}

This distance yields Bayesian updating whenever possible. After a null event, the DM switches to $\mu^*$ and then utilizes Bayes' rule.  This complete belief updating rule was used in \cite{galperti2019persuasion}, and is a special case of both \cite{Myerson1886a, Myerson1886b} and \cite{ortoleva2012}. 

A final example that we wish to mention is the Euclidian distance. 
\begin{defn}[\textbf{Euclidean distance}]\label{euclidian} Let $d_\mu(\pi)=||\mu-\pi||$. Then \[\mu_{E}(s)=\mu(s)+\frac{1-\mu(E)}{|E|}\text{ for any }E\in\Sigma\text{ and }s\in E.\] 
\end{defn}

This distance has several nice features. First, it yields a complete updating rule. Second, the Euclidean distance is a metric, unlike KL divergence. On the other hand, it is always non-Bayesian and ``under utilizes'' prior odds when updating beliefs: probability is allocated to the remaining states (i.e., those in $E$) uniformly. These features echo two consistent findings from experiments: DM's exhibit base-rate neglect \citep{benjamin2019} and are biased toward uniform distributions or the ``ignorance prior" \citep{FoxClemen2005}. 

\section{Axiomatic Characterization}\label{axioms}

In this section, we present three behavioral postulates that characterize \nameref{IEU}. Our first axiom imposes the standard SEU conditions of \cite{anscombeaumann1963} on each conditional preference relation, $\succsim_E$, along with a condition that ensures risk preferences are unaffected by information.  Because these conditions are well-understood, we will not provide a formal discussion of the conditions. 

\begin{axm}[\textbf{SEU Postulates}]\label{SEU} For each $E\in \Sigma$, the following conditions hold.

\begin{itemize}
\item[$(i)$] \textbf{Weak Order:} $\succsim_E$ is complete and transitive. 
\item[$(ii)$] \textbf{Archimedean:} For any $f, g, h\in \F$, if $f\succ_E g$ and $g\succ_E h$, then there are $\alpha, \beta\in (0, 1)$ such that $\alpha f+(1-\alpha) h\succ_E g$ and $g\succ_E \beta f+(1-\beta) h$.
\item[$(iii)$] \textbf{Monotonicity:} For any $f, g\in \F$, if $f(s)\succsim_E g(s)$ for each $s\in S$, then $f\succsim_E g$. 
\item[$(iv)$] \textbf{Nontriviality:} There are $f, g\in \F$ such that $f\succ_E g$. 
\item[$(v)$] \textbf{Independence:} For any $f, g, h\in \F$ and $\alpha\in (0, 1]$, $f\succsim_E g$ if and only if $\alpha f+(1-\alpha) h\succsim_E \alpha g+(1-\alpha) h$. 
\item[$(vi)$] \textbf{Invariant Risk Preference:} For all lotteries $p,q\in\Delta(X)$,  $p\succsim q$ if and only if $p\succsim_E q$.

\end{itemize}
\end{axm}  

The next axiom is standard and ensures that the DM forms a new belief that is consistent with the available information. 

\begin{axm}[\textbf{Consequentialism}]\label{CON}
For any $E \in \Sigma$ and all $f,g\in F$,  \[f(s) =g(s) \text{ for all } s \in E \implies f\sim_{E} g.\]
\end{axm}

The next axiom, \nameref{DCoh}, was introduced in \cite{ortoleva2012}, and a careful discussion may be found there.  In our setting, we say that an event $A$ is \emph{revealed implied by} event $B$ if every state that the DM believes is possible after learning $B$ is also an element of $A$.  \nameref{DCoh} requires that this ``revealed preference'' over events is acyclic. 

\begin{axm}[\textbf{Dynamic Coherence}]\label{DCoh} For any $A_1, \ldots, A_n\subseteq S$, if $S\setminus A_i$ is $\succsim_{A_{i+1}}$-null for each $i\le n-1$ and $S\setminus A_n$ is $\succsim_{A_1}$-null, then $\succsim_{A_1}=\succsim_{A_{n}}$.
\end{axm}

If $S\setminus A_i$ is $\succsim_{A_{i+1}}$-null, then $A_i$ is revealed implied by $A_{i+1}$. Since \nameref{DCoh} implies this relation is acyclic, the revealed preference satisfies SARP. Using the result of \citet{matzkin1991axioms}, an extension of \citet{afriat1967construction} to general budget sets, SARP is a necessary and sufficient condition for the existence of a subjective distance function for belief selection.  

\begin{thm}\label{repthrm} The following are equivalent.

\begin{itemize} 
\item[(i)] A family of preference relations $\{\succsim_E\}_{E\in\Sigma}$ admits an \textbf{\nameref{IEU} representation}. 
\item[(ii)] It satisfies \nameref{SEU}, \nameref{CON}, and \nameref{DCoh}.
\item[(iii)] It admits an \textbf{\nameref{IEU} representation with respect to a continuous, strictly convex distance function}. 
\end{itemize}
\end{thm} 

For a simple intuition behind our result, note that \nameref{SEU} and \nameref{CON} imply that our DM has a conditional belief $\mu_E$ with support contained in $E$, or $\mu_E\in\Delta(E)$. Consequently, we may view each event $E$ as generating a ``budget set,'' $\Delta(E)$, from which the DM must choose her conditional belief. The conditional belief, $\mu_E$, is therefore ``revealed preferred'' to any other belief in the budget set. \nameref{DCoh} ensures that this revealed preference satisfies SARP, allowing for the construction of a ``utility function'' (i.e., a distance function) that generates these beliefs.

Similar to Afriat's theorem, we obtain a continuous, strictly convex distance function without additional restrictions on preferences. The above result holds for an arbitrary collection $\Sigma$ of events. One advantage of our proof is that it is easy to extend to more general models. In \autoref{sect-wiu}, we consider a generalization of \nameref{IEU} that satisfies a weakening of \nameref{CON} and the corresponding characterization theorem uses the same generalization of Afriat's theorem. 

\subsection{Bayesian Updating}

Our main theorem does not require \nameref{DC}, and in fact our axioms are independent of this classic postulate. Similar to results from \cite{Ghirardato2002} and \cite{epstein1993}, imposing \nameref{DC} in our setting ensures that conditional beliefs are consistent with Bayesian updating whenever possible. Recall that $f E h$ denotes that conditional act that returns $f(s)$ for $s \in E$ and $h(s)$ otherwise.

\begin{axm}[\textbf{Dynamic Consistency}]\label{DC} For all non-null events $E \in \Sigma$ and $f, g, h\in \cal F$,  $$f E h \succsim g E h ~\text{ if and only if }~ f\succsim_{E} g.$$
\end{axm}

\begin{prps}\label{prp:DC} A family of preference relations $\{\succsim_E\}_{E\in\Sigma}$ satisfies \textbf{\nameref{SEU}}, \textbf{\nameref{CON}}, \textbf{\nameref{DCoh}}, and \textbf{\nameref{DC}} if and only if it admits an \textbf{{\nameref{IEU}} representation} and $\mu_E=BU(\mu, E)$ for each non-null $E$. 
\end{prps}

Since \nameref{DC} has been discussed extensively, (both \cite{Ghirardato2002} and \cite{epstein1993} include excellent discussions), we will not discuss this result further. Instead, we simply wish to remark that \nameref{DC} places no restrictions on conditional beliefs after null events, which is a major drawback of the standard model. 

A strength of \nameref{IEU} is that it provides a coherent framework for belief revision after null events, which we discuss in \autoref{zero}. Notably, in section \ref{CPSsection} we introduce a strengthening of  \nameref{DC}, which we call \nameref{CC}, that extends the logic of \nameref{DC} to all conditional events and show that this condition characterizes the CPS of \cite{Myerson1886a, Myerson1886b}.

\subsection{Distorted Bayesian Updating}\label{sec:hBayes}

One of the key insights provided by \nameref{IEU} is that distance minimization can be viewed as a unifying framework that accommodates various updating behaviors. In this section, we expand upon this insight by characterizing \nameref{hBayes} and monotonic \nameref{hBayes} with a few simple relaxations of \nameref{DC}.

\begin{axm}[\bf{Consistency}]\label{Consistency} For any non-null $E\in  \Sigma$, $s, s'\in E$, and $x,y \in X$,

\[x \{s\} y \sim x \{s'\} y\text{ implies } x \{s\} y\sim_{E} x \{s'\} y.\]
\end{axm}

\nameref{Consistency} requires that if the DM initially believes that two states are equally likely, then she continues to believe that they are equally likely after observing some event containing them.\footnote{If we strengthen Consistency and the following two axioms by requiring the same condition for null events, we obtain Distorted Bayesian updating with $\delta>0$.} 

We characterize \nameref{hBayes} with one additional condition that we call \nameref{IIA}. This axiom ensures that updating behavior only depends on the probability of a state and not on the name of the state. Further, this condition also ensures that the relative distortions are independent of the realized event.

\begin{axm}[\textbf{Independence of Irrelevant Information}]\label{IIA} For any non-null $E_1,E_2\in  \Sigma \setminus S$, $s, s'\in E_1 \cap E_2$, and $p, q, r\in \Delta(X)$,
\[p \{s\} r \sim_{E_1} q \{s'\}r \text{ if and only if } p \{s\} r \sim_{E_2} q \{s'\}r.\]
\end{axm}

\begin{prps}\label{prp:d-BD} Consider a family of preference relations $\{\succsim_E\}_{E\in\Sigma}$ with an \nameref{IEU} representation. The \nameref{IEU} representation admits a \textbf{\nameref{hBayes} distance} if and only if \textbf{\nameref{Consistency}} and 
\textbf{\nameref{IIA}} hold. 
\end{prps}

We can now characterize Monotonic \nameref{hBayes} distance by introducing a condition ensuring that the DM preserves the ``more likely than'' judgments implied by her prior.

\begin{axm}[\textbf{Monotonicity}]\label{Mon} For any non-null $E\in\Sigma$, $s, s'\in E$, and $x, y\in X$,
\[x \{s\} y\succsim x \{s'\} y\text{ if and only if } x \{s\} y\succsim_{E} x \{s'\} y.\]
\end{axm}

To understand \nameref{Mon},  consider $S=\{s_1,s_2,s_3\}$, $\mu=(12/20, 7/20, 1/20)$, and $E=\{s_2,s_3\}$. Under \nameref{DC}, relative likelihoods are exactly preserved and so a Bayesian DM continues to believe that $s_2$ is seven times as likely as $s_3$ upon learning $E$. Without \nameref{DC}, the \nameref{IEU} would place no restrictions on the conditional relative likelihoods of $s_2$ and $s_3$. Since our DM believed that $E$ was relatively unlikely, it is plausible that she is now less confident in her judgment about the relative odds of $s_2$ and $s_3$. Consequently, she may desire to further modify her belief. For example, she may now think that $s_2$ is only twice as likely as $s_3$, resulting in the posterior $\mu_E=(2/3,1/3)$. Notice that $s_2$ is still more likely than $s_3$; she does not entirely disregard her previous judgments. This restriction is precisely the content of \nameref{Mon}.

\begin{prps}\label{prp:mon} Consider a family of preference relations $\{\succsim_E\}_{E\in\Sigma}$ with an \nameref{IEU} representation. The \nameref{IEU} representation admits a \textbf{Monotonic Distorted Bayesian distance} if and only if \textbf{\nameref{Mon}} and \textbf{\nameref{IIA}} hold. 
\end{prps}



%
%
%
%
%
%

Below we present several examples of \nameref{hBayes} updating. In each of the following examples, we let $S = \{s_1 ,s_2,s_3\}$, and suppose $\mu=(12/20, 7/20, 1/20)$.  In each of the tables, blue (light) shading indicates that the state is under-weighted relative to Bayes' rule, while red (dark) shading indicates the state is over-weighted.

\begin{exm}[Bayesian]\label{BayesUp} Our Distorted Bayesian model includes Bayesian updating as the special case $\delta(x)=x$. These posteriors are given in the table below and will serve as the benchmark to describe our other examples. 

	\begin{table}[h]
		\centering
		\begin{tabular}{ | c | c | c | c |}
			\hline $s / A $ & $\{s_1,s_2\}$ & $\{s_2,s_3\}$ & $\{s_1,s_3\}$ \\
			\hline $s_1$  & $0.63$   & $-$ &  $0.92$ \\
			\hline  $s_2$ & $0.37$   & $0.875$ & $-$ \\  
			\hline  $s_3$ & $-$   & $0.125$ & $0.08$ \\  \hline
		\end{tabular}
		\caption{Bayes' Posteriors for various events}\label{ex1-Bayes}
	\end{table}\end{exm}

\begin{exm}[Under/Over-Reaction] Suppose for some $\alpha>0$, \[\delta(x)= x^{\alpha}.\] Note that for $\alpha=1$ this reduces to Bayes' rule (see \autoref{ex1-Bayes}) . For $\alpha < 1$, the relative probabilities are ``compressed,'' capturing under-reaction to the higher probability state.  One the other hand, when $\alpha > 1$, relative probabilities are ``exaggerated,'' capturing over-reaction to the higher probability state.
%



	\begin{table}[h]
		\centering
		\begin{tabular}{ | c | c | c | c |}
			\hline $s / A $ & $\{s_1,s_2\}$ & $\{s_2,s_3\}$ & $\{s_1,s_3\}$ \\
			\hline $s_1$  & \cellcolor{blue!10}$0.606$   & $-$ &  $  \cellcolor{blue!10}0.88$ \\
			\hline  $s_2$ & $0.394$   & $  \cellcolor{blue!10} 0.826$ & $-$ \\  
			\hline  $s_3$ & $-$   & $0.174$ & $0.12$ \\  \hline
		\end{tabular}
\quad
		\begin{tabular}{ | c | c | c | c |}
			\hline $s / A $ & $\{s_1,s_2\}$ & $\{s_2,s_3\}$ & $\{s_1,s_3\}$ \\
			\hline $s_1$  &  \cellcolor{red!50}$0.656$   & $-$ &  $\cellcolor{red!50}0.95$ \\
			\hline  $s_2$ & $0.344$   & $\cellcolor{red!50}0.912$ & $-$ \\  
			\hline  $s_3$ & $-$   & $0.088$ & $0.05$ \\  \hline
		\end{tabular}
		\caption{Posteriors for $\alpha=0.8$ and $\alpha=1.2$.}\label{ex2}
	\end{table}
Comparing to the Bayesian posteriors in \autoref{ex1-Bayes}, it is simple to see that when $\alpha<1$ the DM always under-weights the more likely state, and when $\alpha>1$ the DM always over-weights the more likely state.

\end{exm} 
\begin{exm}[$S$-reaction] When $\delta$ has a sigmoid shape, it simultaneously captures under-reaction to ``expected states'' and over-reaction to ``unexpected states.'' For some $x_0 \in \R$ and $a >0$,  \[\delta(x)= \frac{1}{1+e^{a(x_0-x)}}.\]

	\begin{table}[h]
		\centering
	\begin{tabular}{ | c | c | c | c |}
			\hline $s / A $ & $\{s_1,s_2\}$ & $\{s_2,s_3\}$ & $\{s_1,s_3\}$ \\
			\hline $s_1$  & \cellcolor{red!50}$0.69$   & $-$ & \cellcolor{blue!10}$0.91$ \\
			\hline  $s_2$ & $0.31$   & \cellcolor{blue!10}$0.821$ & $-$ \\  
			\hline  $s_3$ & $-$   & $0.179$ & $0.09$ \\  \hline
		\end{tabular}
		\caption{Posteriors for $S$-reaction with $a=6, x_0=0.5$.}\label{ex3}
	\end{table}

Compared to the Bayesian posteriors, the DM over-weights $s_1$ after $\{s_1, s_2\}$, exhibiting features of over-reaction, while the DM under-weights $s_1$ after $\{s_1, s_3\}$ and  $s_2$ in after $\{s_2, s_3\}.$ This is because the shape of $\delta$ induces over-reaction to rare events, thereby increasing the probability of $s_3$. 

\end{exm}

\begin{exm}[Confirmation Bias]\label{conf} Confirmation bias refers to the tendency to give extra credence to ``believed hypothesis.'' For some $b > 0$, let  \[\delta(x)= x + b\, \mathds{1}\left\{x > \frac12 \right\}. \]  Under this rule, states which are believed to be more likely are biased by $b$. 

\begin{table}[h]
		\centering
	\begin{tabular}{ | c | c | c | c |}
			\hline $s / A $ & $\{s_1,s_2\}$ & $\{s_2,s_3\}$ & $\{s_1,s_3\}$ \\
			\hline $s_1$  &  \cellcolor{red!50} $\frac{12+20b}{19+20b}$   & $-$ &   \cellcolor{red!50} $\frac{12+20b}{13+20b}$ \\
			\hline  $s_2$ & $\frac{7}{19+20b}$   & $.875$ & $-$ \\  
			\hline  $s_3$ & $-$   & $.125$ & $\frac{1}{13+20b}$ \\  \hline
		\end{tabular}
		\caption{Posteriors under Confirmation Bias.}\label{ex3}
	\end{table}
The DM always over-reacts to  $s_1$, her favored state, whenever information allows. When the information precludes $s_1$ she behaves in accordance with Bayes' rule. 
\end{exm}

\subsection{Other Forms of Non-Bayesian Updating}\label{sec:nonBayes}

There are of course many forms of non-Bayesian updating captured by \nameref{IEU} that fall outside of \nameref{hBayes}.  Below we illustrate how the \nameref{mixedbayes} distance can capture motivated reasoning and wishful thinking. 

\begin{exm}[Mixed Bayesian Optimism]\label{hist} We still let $S = \{s_1 ,s_2,s_3\}$ and suppose $\mu=(12/20, 7/20, 1/20)$, as before. Now suppose our DM uses the \nameref{mixedbayes} distance with $\rho = (0, 0, 1)$, where $\rho$ captures the idea that $s_3$ is the ``best state,'' i.e., the DM has a motivation to believe that $s_3$ is true.  

After $\{s_1,s_2\}$ is realized, the posteriors are identical to the Bayesian posteriors because $s_3$ has been ruled out. For the other two events, the DM exhibits ``reversals.'' Under both $\{s_2, s_3\}$ and $\{s_1,s_3\}$ the DM believes $s_3$ is now the most likely state, which violates \nameref{Mon}. The belief after $\{s_2,s_3\}$ is more extreme because $\{s_2,s_3\}$ is ``unexpected'' under the prior, which pushes the DM more toward $\rho$. 

	\begin{table}[h]
		\centering
		\begin{tabular}{ | c | c | c | c |}
			\hline $s / A $ & $\{s_1,s_2\}$ & $\{s_2,s_3\}$ & $\{s_1,s_3\}$ \\
			\hline $s_1$  & $0.63$   & $-$ &  $0.36$ \\
			\hline  $s_2$ & $0.37$   & $0.25$ & $-$ \\  
			\hline  $s_3$ & $-$   & $0.75$ & $0.64$ \\  \hline
		\end{tabular}
		\caption{Posteriors under Mixed Bayesian updating}\label{ex1-MBayes}
	\end{table}

\end{exm}

%

\section{Updating After Zero-probability Events}\label{zero}

The most well-known limitation of Bayesian updating is that it is incomplete; it is not defined for zero-probability events. This is particularly problematic in game theoretic settings, where beliefs are induced by the equilibrium strategies and any action off the equilibrium path leads to a zero-probability event. In contrast, our notion of belief updating is well-defined for zero-probability events. Thus, \nameref{IEU} provides a way to extend (non-)Bayesian updating to all events. 

\subsection{Conditional Probability System}\label{CPSsection}

Perhaps the most well-known method for handling beliefs conditional on null-events is the conditional probability system (CPS) introduced by \cite{Myerson1886a, Myerson1886b}.\footnote{The idea of CPS goes back to \cite{renyi1955new}.} The development of CPS is closely related the developments of Perfect Bayesian Equilibrium and its refinements.  PBE requires that agents's beliefs are Bayes-consistent with the prior whenever possible. However, PBE does not make any restrictions when Bayes' rule is not applicable. Hence, PBE may allow for some unreasonable beliefs after actions off the equilibrium path. The Sequential Equilibria of \cite{KrepsWilson1982} refines the PBE by requiring that any belief in sequential equilibria should be a limit of full-support beliefs after applying Bayes rule accordingly. Checking whether conditional beliefs can be supported by full-support beliefs is not an easy task, and \cite{Myerson1886a, Myerson1886b} shows that this limit requirement of sequential equilibria is equivalent to the following simple condition. 

\begin{defn}\label{cps} A \textbf{Conditional Probability System} (CPS) is a collection $\{\mu_{E}\}_{E\in\Sigma}$ of conditional probability distributions such that for all $s \in F \subseteq E$, 
\begin{equation}\label{CPSeq}\mu_E(s)=\mu_F(s)\, \mu_E(F).\end{equation}
\end{defn}
When $\mu_E(F)\neq 0$, \autoref{CPSeq} reduces to Bayes' rule. However, when $\mu_E(F)=0$, it implies that $\mu_E(s)=0$ as well, and so it places no restriction directly on $\mu_F(s)$.

As we will show below, CPS is a special case of our model. A major distinction between CPS and \nameref{IEU} is that CPS requires Bayesian updating whenever possible, while   \nameref{IEU} provides a unifying framework that allows for Bayesian and non-Bayesian updating. To characterize CPS, we introduce the following strengthening of \nameref{DC}. 

\begin{axm}[\textbf{Conditional Consistency}]\label{CC} For all $E \in \Sigma$, $\succsim_{E}$-non-null $A \subset E$, and $f, g, h\in \cal F$,  $$f\,A\, h \succsim_{E} g\,A\, h ~\text{ if and only if }~ f\succsim_{A} g.$$
\end{axm}

\nameref{CC} implies \nameref{DC} but also has bite on events that are $\succsim$-null. In essence, \nameref{CC} extends the logic of \nameref{DC} to all conditional preferences $E$ and nested events that are $\succsim_E$-non-null.

To illustrate \nameref{CC}, imagine a coin flip. The states $h$ and $t$ are the usual outcomes of heads or tails, $e$ and $e'$ denote edges where $e'$ has been warn thin, while $l_1$ and $l_2$ denote landing on a marked location, which yields the state space $S=\{h, t, e, e', l_1, l_2\}$.  Initially, the DM believes that $\mu(h)=\mu(t)=\frac12$, and treats the other states as null. 

Suppose the DM is informed that, astonishingly, the coin did not land on a face; $A=\{e, e', l_1, l_2 \}$ was realized.   Further, suppose that our DM believes that the coin landing on either of the marked locations is \emph{more impossible} than its landing on an edge. Accordingly, her conditional beliefs are $\mu_A(\{e, e'\})=1$ and $\mu_A(\{l_1, l_2\})=0$. If this information is further refined so that $\{e\}$ is ruled out and our DM continues to utilize Bayes' rule, then we expect $\mu_B(\{e'\})=1$ (where $B=\{e', l_1, l_2\}$).  \nameref{CC} imposes Dynamic Consistency between $\mu_A$ and $\mu_B$ because $B$ becomes $\succsim_A$-non-null and $B \subset A$. 

Our next theorem states that \nameref{CC} is the precise strengthening of \nameref{DC} required to characterize CPS. 

\begin{thm}\label{CPSthrm} A family of preference relations $\{\succsim_E\}_{E\in\Sigma}$ satisfies \textbf{\nameref{SEU}}, \textbf{\nameref{CON}}, and \textbf{\nameref{CC}} if and only if it admits a \textbf{CPS representation}.
\end{thm} 

While our theorem ensures that the collection of beliefs satisfies the requirement of a CPS (\autoref{cps}), it does not directly shed light on the structure of the CPS. It does not imply yet that CPS is a special case of  \nameref{IEU}.

 Our next proposition shows that any CPS is a special case of \nameref{IEU} and it can be described by a collection of beliefs whose supports partition $S$. Further, the DM moves between these beliefs in an ``ordered'' fashion and this CPS representation is generated by a support-dependent bayesian distance. 

\begin{prps}\label{CPS=MD} Suppose a family of preferences $\{\succsim_{E}\}_{E\in \Sigma}$ admits a \textbf{CPS representation}. Then
there are $\mu^0, \ldots, \mu^K\in\Delta(S)$ such that $\text{sp}(\mu^0), \ldots, \text{sp}(\mu^K)$ is a partition of $S$ and for any $E\in\Sigma$, \[\mu_{E}=\text{BU}(\mu^{k^*}, E)\text{ where }k^*=\min\{k \mid \mu^k(E) > 0\}.\]

\noindent Moreover, $\{\succsim_{E}\}_{E\in \Sigma}$ has an \textbf{\nameref{IEU} representation} with respect to the following distance function: 
\[d_\mu(\pi)=\beta^\sigma(\mu^{k^*}_i, \pi) +k^*\,\big(\sigma(1)+|\sigma(0)|\big)\,\]
where $k^*=\min\{k \mid \mu^k( \text{sp}(\pi))>0\}$.\footnote{The first part of this proposition is not entirely new. \cite{KrepsWilson1982} already pointed out a connection between sequential equilibria beliefs and a collection of linearly ordered priors $\mu^0, \ldots, \mu^K$.}
\end{prps}

Note that Proposition \ref{os2} is a special case of the above result when $K=1$. 

\begin{exm}[Coin Flip]\label{coin}
Recall the coin flip example from before, where the states are $S=\{h, t, e, e', l_1, l_2\}$, where $h$ and $t$ correspond to heads or tails, $e$ and $e'$ correspond to the coin landing on an edge, where one edge is thinner than the other, while $l_1$ and $l_2$ correspond to the coin landing on precisely marked locations. These possibilities are described by the probability distributions 
\[ \mu^0(s)=\begin{cases} \frac12 & s \in \{h, t\} \\ 0 & \text{otherwise}\end{cases}; \mu^1(s)=\begin{cases} \frac78 & s = e \\ \frac18 & s =e' \\ 0 & \text{otherwise}\end{cases};\text{ and }\mu^2(s)=\begin{cases} \frac12 & s \in \{l_1, l_2\} \\ 0 & \text{otherwise}\end{cases}. \]

Our DM has the initial prior $\mu^0$ (i.e., $\succsim$ has an SEU representation with $(u, \mu^0)$). Suppose she observes $A=\{e, e', l_1, l_2\}$. Since $\mu^0(A)=0$, Bayesian updating is not defined. After $A$, the DM selects $\mu^1$ (i.e., $\mu_A=\mu^1$) because it is of ``lower order'' than $\mu^2$ and therefore it takes precedence. 
\end{exm}

\subsection{Hypothesis Testing}\label{HTsection}

A recent and elegant addition to the literature on updating after zero-probability events is the Hypothesis Testing model (HT) of \cite{ortoleva2012}.  Such an agent will update using Bayes' rule for expected events: events with probability above some threshold $\epsilon$. When an event $E$ is unexpected (i.e., under the agent's prior $\mu(E)\le\epsilon$), the agent rejects her prior, updates a second-order prior over beliefs, and selects a new belief according to a maximum likelihood procedure. Formally, a HT representation is given by a triple, $(\mu,\rho,\epsilon)$, consisting of a prior $\mu \in \Delta(S)$, a second order prior $\rho\in\Delta(\Delta(S))$, and a threshold $\epsilon\in [0, 1)$ with the requirement that $\mu=\arg\max_{\pi\in \Delta(S)} \rho(\pi)$. Then, for any $E\in\Sigma$, 
\[\mu_{E}=
\begin{cases}
\text{BU}(\mu, E)&\text{ if }\mu(E)> \epsilon,\\
\text{BU}(\pi^{E}, E)&\text{ otherwise}.\end{cases}
\]
where $\pi^{E}=\arg\max_{\pi\in \Delta(S)} \rho(\pi) \pi(E)$. It turns out that HT is behaviorally equivalent to \nameref{IEU}. 

\begin{cor}\label{HT=IEU} A family of preference relations $\{\succsim_E\}_{E\in\Sigma}$ admits an \textbf{HT representation} if and only if it admits an \textbf{\nameref{IEU} representation}. 
\end{cor}

This corollary follows from our \autoref{repthrm} and Theorem 1 of \cite{ortoleva2012}.  However, it is important to note that our proof techniques are quite different.

\subsection{Relating HT and CPS}

The formal relationship between HT and CPS has not previously been established. Our results, \autoref{HT=IEU} and \autoref{CPS=MD}, indirectly show that CPS is a special case of HT. Further, since every CPS satisfies Bayes' rule, it is a special case of HT with $\epsilon=0$. 

\begin{cor}\label{CPS-HTE} If a family of preference relations $\{\succsim_E\}_{E\in\Sigma}$ admits an \textbf{CPS representation}, then it admits an \textbf{HT representation} with $\varepsilon=0$.\end{cor}

However, the converse does not hold; even when $\epsilon =0$, HT preferences may be inconsistent with CPS preferences. The reason for this is due to the way in which the selection of new beliefs occurs in HT. Indeed, our previous results formally show why. In HT, $\epsilon=0$ if and only if Dynamic Consistency holds. Hence, Proposition 3 characterizes HT with $\epsilon=0$. Theorem 2 implies that HT with $\epsilon=0$ is strictly more general than CPS since Conditional Consistency is strictly stronger than Dynamic Consistency.

\subsection{A Non-Bayesian CPS}\label{eCPS}

A natural way to generalize the CPS is to retain the sequential selection of new beliefs while incorporating the idea of ``non-Bayesian reaction to unexpected events'' from the HT model. To do so, we introduce $\epsilon$-CPS, a one-parameter, non-Bayesian extension of the CPS. This extension may lead to an interesting, non-Bayesian generalization of sequential equilibria. 

\begin{defn}A family of preferences $\{\succsim_{E}\}_{E\in \Sigma}$ admits an $\epsilon$-\textbf{CPS} representation if there are probability distributions $\mu^0, \ldots, \mu^K \in \Delta(S)$ and $\epsilon\in [0, 1)$ such that 
\[\mu_E=\text{BU}(\mu^{k^*}, E)\text{ where }k^*=\min\{k\le K \mid \mu^k(E)>\epsilon\},\] 
for every $E\in\Sigma$. 
\end{defn}

The $\epsilon$-CPS representation incorporates the key idea of HT by allowing for non-Bayesian reactions to unexpected events:  $\mu^k(E)\leq\epsilon$. However, it provides additional structure to the posterior selection process. The $\epsilon$-CPS remains a special case of HT and \nameref{IEU}.

\begin{thm}\label{eCPS} Any  \textbf{$\epsilon$-CPS representation} also has a \textbf{HT representation}. Moreover, if $\epsilon=0$, then the threshold for the HT representation is also zero. 
\end{thm}

\subsection{Relationships}

Since there are multiple approaches to updating after zero probability events, we summarize their relationship to each other and the key axioms in \autoref{modelrelations}.

\begin{figure}[h]
\centering
\includegraphics[width=10cm]{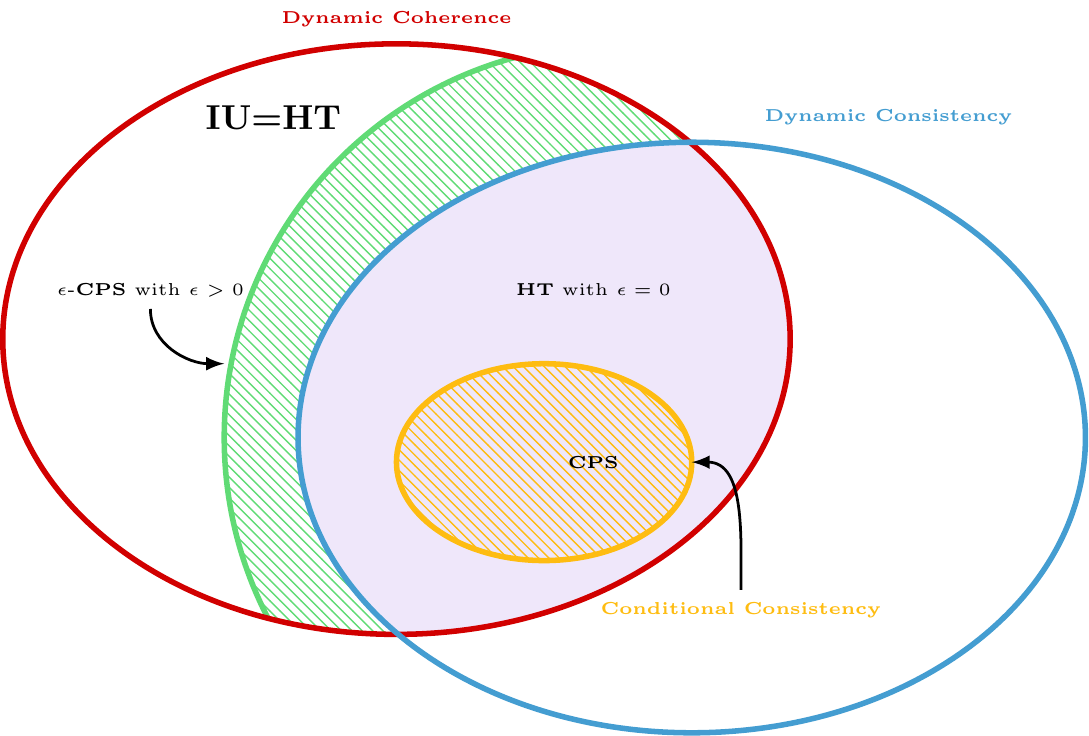}
     \caption{Relationship between complete updating models.}\label{modelrelations}
         \end{figure}
         
\autoref{modelrelations} clearly illustrates two notable discoveries. First, \nameref{CC} implies both \nameref{DC} and \nameref{DCoh}.  Second, Dynamically Consistent HT is strictly more general than the CPS.

\section{Incorporating a Signal Structure}\label{Signal}

While our setting is quite general, it is often useful to make explicit reference to a signal structure.  We therefore illustrate that our framework can incorporate standard signal structures utilized in experimental settings and game theory by introducing more structure to the state space (e.g., $S$ has a product structure). 

Let $\Omega$ be the payoff relevant state space and $M$ be the set of all signals. For each $\omega\in \Omega$ and $m\in M$, let $P(\omega)$ be the (unconditional) probability that the payoff relevant state $\omega$ occurs and $P(m|\omega)$ be the (conditional) probability that the DM receives the signal $m$ when the state is $\omega$. Indeed, receiving a signal is equivalent to observing an event in an expanded state space, $S=\Omega\times M$. Specifically, receiving the signal $m$ is equivalent to observing the event $\{(\omega, m)\}_{\omega\in \Omega}$ in $S$. 

Let $\mu$ be the prior on $S$, so that $\mu_{\omega m}=P(m|\omega)\,P(\omega)$ for each $(\omega, m)$. In the case of Bayesian updating, the connection between our framework and the signal structure is straightforward. Note that the Bayesian divergence generates Bayesian updating in the signal structure framework:
\[P(\omega|m)=\frac{\mu_{\omega m}}{\sum_{\omega'\in \Omega}\mu_{\omega' m}}=\frac{P(m|\omega)\,P(\omega)}{\sum_{\omega'\in \Omega}P(m|\omega')\,P(\omega')}.\]

A similar connection is possible for non-Bayesian updating rules. For example, consider the following distance function. For $\alpha,\beta \ge 0$,
\[d_\mu(\pi)=\sum_{(\omega, m)\in \Omega\times M} \left(\sum_{m'\in M}\mu_{\omega m'}\right)^{\alpha-\beta}\,\mu^{\beta}_{\omega m}\log\left(\frac{\pi_{\omega m}}{\mu_{\omega m}}\right).\]
This distance generates the posterior
\[P(\omega|m)=\frac{\big(\sum_{m'\in M}\mu_{\omega m'}\big)^{\alpha-\beta}\,\mu^\beta_{\omega m}}{\sum_{\omega'\in \Omega}\big(\sum_{m'\in M}\mu_{\omega' m'}\big)^{\alpha-\beta}\,\mu^\beta_{\omega' m}}=\frac{(P(m|\omega))^\beta\,(P(\omega))^\alpha}{\sum_{\omega'\in \Omega}(P(m|\omega'))^\beta\,(P(\omega'))^\alpha},\]
which is precisely the non-Bayesian updating rule proposed by \cite{grether1980}. This is a simple generalization of Bayes' rule, where $\alpha$ captures the influence of the prior and $\beta$ captures the influence of the signals.

In general, the following distance function
\[d_\mu(\pi)=\sum_{(\omega, m)\in \Omega\times M}g\left(\sum_{m'\in M} \mu_{\omega m'}\right)f\left(\frac{\mu_{\omega m}}{\sum_{m'\in M} \mu_{\omega m'}}\right)\log\left(\frac{\pi_{\omega m}}{\mu_{\omega m}}\right)\]
generates Distorted Bayesian updating in the signal structure framework:
\[P(\omega|m)=\frac{f\big(\sum_{m'\in M} \mu_{\omega m'}\big)g\big(\frac{\mu_{\omega m}}{\sum_{m'\in M} \mu_{\omega m'}}\big)}{\sum_{\omega'\in \Omega} f\big(\sum_{m'\in M} \mu_{\omega' m'}\big)g\big(\frac{\mu_{\omega' m}}{\sum_{m'\in M} \mu_{\omega' m'}}\big)}=\frac{f\big(P(m|\omega)\big)g\big(P(\omega)\big)}{\sum_{\omega'\in \Omega}f\big(P(m|\omega')\big)g\big(P(\omega')\big)}.\]
The above updating rule reduces to Grether's rule when $f(x)=x^\beta$ and $g(x)=x^\alpha.$ We apply this updating rule to ``Bayesian'' persuasion games in \autoref{sec-app}.

\begin{exm} Consider the following example, with $\Omega=\{\omega_H, \omega_L\}$ and $M=\{h,l\}$. We suppose $P(\omega_H)=\frac58$ and $P(h | \omega_H)=P(l | \omega_L)=\frac35$.

	\begin{table}[h]
		\centering
		\begin{tabular}{ | c | c | c |}
			\hline $$ & $h$ & $l$ \\
			\hline $\omega_H$  & $0.375$   & $0.25$   \\
			\hline  $\omega_L$ & $0.15$   & $0.225$ \\  \hline
		\end{tabular}
		\quad
			\begin{tabular}{ | c | c | c | }
			\hline  & $\mu( \cdot | h ) $ & $\mu( \cdot | l ) $  \\
			\hline $\omega_H$  &  $0.7143$  & $0.5263$ \\
			\hline  $\omega_L$ &  $0.2857$  & $0.4737$  \\ \hline 
		\end{tabular}
		\caption{Induced prior $\mu$ over $S=\Omega \times M$ and the corresponding Bayes' posteriors.}\label{ex-signal}
	\end{table}

\autoref{ex-signal} Illustrates the prior over $S$ and the resulting posterior beliefs under Bayesian updating.  Applying the Distorted Bayesian distance yields the following conditional probabilities for $\omega_H$ after signals $h$ and $l$:

\[ \mu( \omega_H | h ) = \frac{f(0.6)g(0.625)}{ f(0.6)g(0.625) + f(0.4)g(0.375)},\] 
\[ \mu( \omega_H | l ) = \frac{f(0.4)g(0.625)}{ f(0.4)g(0.625) + f(0.6)g(0.375)}.\] 

To further illustrate, we consider several specifications for $f(x)=x^\beta$ and $g(x)=x^\alpha$ in the table below. 

	\begin{table}[h]
		\centering
		\begin{tabular}{ | c | c | c | }
			\hline  & $\mu( \cdot | h ) $ & $\mu( \cdot | l ) $  \\
			\hline $\omega_H$  &  $ \cellcolor{red!50}0.7264$   &  \cellcolor{blue!10}$0.4603  $  \\
			\hline  $\omega_L$ & $0.2736$   & $0.5397$\\
		 \hline 
		\end{tabular}
		\quad
		\begin{tabular}{ | c | c | c | }
			\hline  & $\mu( \cdot | h ) $ & $\mu( \cdot | l ) $  \\
			\hline $\omega_H$  &  \cellcolor{blue!10}$0.6755$   & \cellcolor{blue!10}$0.5211$  \\
			\hline  $\omega_L$ & $0.3245$   & $ 0.4789$\\
		 \hline 
		\end{tabular}
		\quad
			\begin{tabular}{ | c | c | c | }
			\hline  & $\mu( \cdot | h ) $ & $\mu( \cdot | l ) $  \\
			\hline $\omega_H$  &  \cellcolor{red!50}$0.7347$  & \cellcolor{red!50}$0.5517$ \\
			\hline  $\omega_L$ &  $0.2653$  & $0.4483$  \\
			 \hline 
		\end{tabular}
		\caption{The left table reports posteriors with Base-rate neglect ($\alpha =0.8$) and over-reaction to signals ($\beta=1.4$), the middle table reports posteriors with Base-rate neglect ($\alpha =0.8$) and under-reaction to signals ($\beta=0.8$), and the right table reports posteriors with Base-rate bias ($\alpha =1.2$) and an accurate reaction to signals ($\beta=1$).} 
	\end{table}

\end{exm}

\section{Application to Bayesian Persuasion}\label{sec-app}

In this section, we demonstrate the usefulness of our model by applying it to the Bayesian persuasion games of \cite{kamenica2011bayesian}. In particular, we analyze the effects of non-Bayesian updating rules on the optimal information structure. We first describe the general Bayesian persuasion environment. Let $\Omega$ be the set of payoff-relevant states and $\rho$ be a prior over $\Omega=\{\omega_1, \ldots, \omega_n\}$. Let $A$ and $M$ be the finite sets of actions and messages, respectively. A signal structure is a function $\pi:\Omega \to \Delta(M)$. Given action $a\in A$ and state $\omega\in\Omega$, the receiver's payoff is $u(a, \omega)$ and the sender's payoff is $v(a, \omega)$. Given message realization $m$ and signal structure $\pi$, the receiver's optimal action is determined by
 \[a^*_\pi(m)=\arg\max_{a\in A} \mathbb{E}_{\omega\sim \mu_\pi(\cdot|m)}u(a, \omega)\]
 where $\mu_\pi(\cdot|m)$ is a conditional probability distribution over $\Omega$. The sender's goal is to persuade the receiver to take certain actions by choosing a signal structure $\pi$. The optimal signal structure for the sender must solve
 \[\max_{\pi\in\Pi } \mathbb{E}_{\omega\sim\rho} \mathbb{E}_{m\sim \pi(\omega)} v\big(a^*_\pi(m), \omega\big).\]
 
To illustrate the implications of our model, we now consider a simpler environment with two actions;  $A=\{a, b\}$. Then
 \[a^*_\pi(m)=a\text{ if }\sum^n_{i=1}\mu_\pi(\omega_i|m)\, u_i\ge 0,\]
where $u_i=u(a, \omega_i)-u(b, \omega_i)$. We assume that the sender always prefers action $a$, which is captured by $v(a, \omega)=1$ and $v(b, \omega)=0$. Hence, the sender maximizes
 \[\sum_{m\in M} \sum_{\omega\in\Omega}\mu(\omega, m)\mathds{1}\{a^*(m)=a\}.\]
This simple environment is rich enough to nest the judge-prosecutor example of \cite{kamenica2011bayesian} and the police-driver example of \cite{kamenica2019bayesian}. Since $\min\{|\Omega|, |A|\}=2$, we will first assume that $|M|=2$ and consider the case of $|M|\ge 3$ in online \autoref{rich}.\footnote{When $f$ is not linear, the revelation principle may be violated (see \cite{de2022non}). Hence, the assumption $|M|=2$ is not without loss of generality. We show that our main findings do not change substantively when $|M|\ge 3$ (see online \autoref{rich}). }
   
To apply our model, let $S=\Omega\times M$ and let $\mu$ be a prior over $S$ determined by $\rho$ and $\pi$: $\mu(\omega_i, m)=\rho_i\, \pi_m(\omega_i)$.\footnote{Bayesian plausibility is already satisfied with this Cartesian structure.} Our model determines the conditional probability $\mu_\pi(\cdot|m)$ and the rest is standard. We assume biased Bayesian updating defined in \autoref{Signal}, where the conditional probability is given by
 \[\mu_\pi(\omega_i|m)=\frac{g\big(\rho_i\big)\,f\big(\pi_m(\omega_i)\big)}{\sum^n_{j=1}g\big(\rho_j\big)\,f\big(\pi_m(\omega_j)\big)}.\]
 We assume that $f$ and $g$ are positive valued and $f$ is strictly increasing.\footnote{Although it is not essential, for simplicity, we assume $f$ is differentiable.} Below we demonstrate how the curvature of $f$ determines the form of the optimal signal structure for the sender. We find that as it distorts prior probabilities $g$ has no qualitative impact on the optimal signal structure, whereas the optimal signal structure significantly varies with the curvature of $f$. In particular, the set of states at which the sender is fully revealing when $f$ is concave is drastically different from the set of fully reveling states when $f$ is strictly convex.

 Given this updating rule, the optimal action of the receiver is 
 
 \[a^*_\pi(m)=a\text{ iff }\sum^n_{i=1}g\big(\rho_i\big)f\big(\pi_m(\omega_i)\big)\, u_i\ge 0.\]
The sender's optimization problem is 
\[\max_{\pi\in \Pi} V(\pi)=\sum_{m\in M}\Big(\sum^n_{i=1} \rho_i\, \pi_m(\omega_i)\Big)\mathds{1}\big\{\sum^n_{i=1}g\big(\rho_i\big)f\big(\pi_m(\omega_i)\big)\, u_i\ge 0\big\}.\]
 
To simplify the exposition, we first rule out some uninteresting scenarios in which persuasion does not matter. Note that the maximum value for $V$ is $1$. To focus on the interesting cases, suppose now that we have $u_1, \ldots, u_n$ and $\rho$ such that $V=1$ cannot be achieved. This assumption implies that the fully revealing signal structure is not optimal; i.e., $\sum^n_{i=1}g\big(\rho_i\big)\, u_i<0$. Since messages $m_1$ and $m_2$ are symmetric, we will focus on signal structures such that $a^*_\pi(m_1)=a$. 

Let  $A=\{i\le n: u_i\ge 0\}$ denote the set of states in which the sender's and receiver's interests are aligned.  Then the sender's problem is simply to maximize 
\[\max_{\pi_{m_1}\in [0, 1]^{|\Omega|}}\sum^n_{i=1} \rho_i\, \pi_{m_1}(\omega_i)\text{ subject to }\sum_{i\in A}g\big(\rho_i\big)f\big(\pi_m(\omega_i)\big)\, u_i\ge \sum_{i\in A^c}g\big(\rho_i\big)f\big(\pi_m(\omega_i)\big)\, |u_i|.\]
Intuitively, the sender must optimally allocate the utility generated from the states in $A$ (i.e., $\sum_{i\in A}g\big(\rho_i\big)f\big(\pi_m(\omega_i)\big)\, u_i)$) across the states in $A^c$. Since the objective function is linear, the curvature of $f$ essentially dictates the form of the optimal signal structure. 

\begin{prps}\label{concave} Suppose either $f(x)=x$ and $\frac{g(\rho_i) u_i}{\rho_i}\neq \frac{g(\rho_j) u_j}{\rho_j}$ for any $i, j$ with $u_i, u_j<0$ or $f$ is strictly concave. For any optimal signal structure $\pi^*$, there is $\bar{\omega}\in A^c$ and non-empty $\Omega_1\supseteq A$ such that 
\[\pi^*_{m_1}(\omega_1)=1\text{ for any }\omega_1\in\Omega_1\text{ and }\pi^*_{m_1}(\omega_2)=0\text{ for any }\omega_2\in\Omega\setminus (\Omega_1\cup\{\bar{\omega}\}).\]   
\end{prps}

\autoref{concave} shows that the sender randomizes at no more than one state (i.e., $\bar{\omega}$) when $f$ is concave. That is because when $f$ is concave, the sender's objective function is convex. Hence, the optimal signal structure is essentially an extreme point of $[0, 1]^{\Omega}$, ignoring $\bar{\omega}$.\footnote{\cite{de2022non} show that optimal signal structures in the special case of Grether's rule with $f(x)=x$ and $g(x)=x^\alpha$ are not qualitatively different from the standard Bayesian case. Our proposition shows a similar result in this different environment.}

However, when $f$ is not concave, in particular when $f'(0)=0$, extreme points of $[0, 1]^{\Omega\setminus\{\bar{\omega}\}}$ cannot be optimal. In fact, the sender randomizes at as many states in $A^c$ as possible, depending on the total resource generated by states in $A$.

\begin{prps}\label{convex} Suppose $f'(0)=0$.\footnote{The strict convexity of $f$ is not necessary for this result.} For any optimal signal structure $\pi^*$, there is $\Omega_1\supseteq A$ such that 
\[\pi_{m_1}(\omega_1)=1\text{ for any }\omega_1\in\Omega_1\text{ and }\pi_{m_1}(\omega_2)\in (0, 1)\text{ for any }\omega_2\in\Omega\setminus \Omega_1.\]   
\end{prps} 

\autoref{convex} shows that for Grether's $\alpha-\beta$ rule, the optimal signal structure in the case of $\beta\le 1$ (including Bayesian updating) is qualitatively different from the case of $\beta>1$. The difference between the cases $\beta\le 1$ and $\beta>1$ is more precisely illustrated by following example. 

\begin{exm}\label{persuasion} Suppose $|\Omega|=3$. Let $\rho=(\frac{1}{7}, \frac{3}{7}, \frac{3}{7})$ and $(u_1, u_2, u_3)=(1, -\frac{1}{2}, -1)$. When $\beta\le 1$, 
\[\pi^H_{m_1}=1,\,\pi^M_{m_1}=\Big(\frac{2}{3}\Big)^\frac{1}{\beta}\in (0, 1), \text{ and }\pi^L_{m_1}=0.\]  
However, when $\beta>1$, 
\[\pi^H_{m_1}=1,\, \pi^M_{m_1}=\frac{2^\frac{1}{\beta-1}}{\Big(3(1+2^{\beta-1})\Big)^\frac{1}{\beta}}\in (0, 1), \text{ and }\pi^L_{m_1}=\frac{1}{\Big(3(1+2^{\beta-1})\Big)^\frac{1}{\beta}}\in (0, 1).\] 
\end{exm}

\section{Partial Consequentialism and Weighted \nameref{IEU}}\label{sect-wiu} 

In this section, we generalize our main result by relaxing \nameref{CON}. Following our analogy to revealed preference theory, \nameref{CON} ensures that $E$ is equivalent to the budget set $\Delta(E)$. By dropping \nameref{CON}, we allow for the DM to perceive a subjective budget set from which she may choose. For instance, this may be because the DM perceives the information as less reliable than the analyst, or the DM may have an imperfect memory and her uncertainty about which event transpired is reflected in her beliefs.  We do however impose two natural conditions on her behavior. 

\begin{defn}[{\textbf{wIU}}]\label{WIEUdef} A family of preference relations $\{\succsim_E\}_{E\in\Sigma}$ admits a \textbf{Weighted Inertial Updating} representation if there are a Bernoulli utility function $u:X\to\mathds{R}$, a prior $\mu \in \Delta(S)$, a distance function $d_\mu:\Delta(S)\to\mathds{R}$, and a weight $\gamma \in [0, 1)$ such that for each $E \in \Sigma$, the preference relation $\succsim_E$ admits a SEU representation with $(u, \mu_E)$, where
\begin{equation}\mu_E\equiv \gamma\,\mu+(1-\gamma)\,\argmin_{\pi \in \Delta(E)}  d_\mu(\pi).\end{equation}
\end{defn}

This generalization of \nameref{IEU} nests the updating rules studied in \cite{Epstein2006}, \cite{kovach2020},  and \cite{Epstein2008}. 

We first demonstrate that \nameref{IEU} representations can be generated from \nameref{WIEUdef} representations by imposing \nameref{CON}. 

\begin{prps} If a family of preference relations $\{\succsim_E\}_{E\in\Sigma}$ admits a \nameref{WIEUdef} representation and satisfies \nameref{CON}, then it also admits an \nameref{IEU} representation. 

\end{prps}

To characterize \nameref{WIEUdef}, we need to weaken \nameref{DCoh} and \nameref{CON} to accommodate the DM's partial reaction to information. While our DM does not fully incorporate the informational content of the event $A$, her belief in $A$ increases and, consequently, she necessarily gives lower credence to $S\setminus A$ and any $E \subseteq S \setminus A$. 

We introduce the following definition to capture the DM's subjective perception of events that become relatively less likely after $A$. 

\begin{defn}[Unfavored Event] We say $E$ is $\succsim_A$-unfavored if for any $E'\subseteq S$ and $p, q\in \Delta(X)$, 
\[p\,E'\,w\sim q\,E\,w\text{ implies }p\,E'\,w\succsim_A q\,E\,w,\]
with at least one strict inequality for some $E'$. We then say $E$ is a $\succsim_A$-favored event if $E^c$ is $\succsim_A$-unfavored.
\end{defn}

Similar to how \nameref{DCoh} ensures a consistent reaction to null events, \nameref{CDCoh} ensures a consistent reaction to favored events.

\begin{axm}[\textbf{Partial Dynamic Coherence}]\label{CDCoh} For any $A_1, \ldots, A_n\subseteq S$, if $A_i$ is $\succsim_{A_{i+1}}$-favored for each $i\le n-1$ and $A_n$ is $\succsim_{A_1}$-favored, then $\succsim_{A_1}=\succsim_{A_{n}}$.
\end{axm}


Next, we require that her subjective belief in $E$ weakly increases after she is told that $E$ has occurred. While \nameref{CON} demands that the DM is convinced of $E$, our novel axiom, \nameref{PC}, only demands that she puts more stock in $E$. 

\begin{axm}[\textbf{Partial Consequentialism}]\label{PC} For any $E\subseteq S$, $E$ is $\succsim_E$-favored.
\end{axm}

Finally, we require a condition to ensure a consistent reaction to all events. That is, the following condition guarantees that $\gamma$ is event independent.

\begin{axm}[\textbf{Relative Tradeoff Consistency}]\label{RTC} For any $A, B\in\Sigma$, $p, q, r\in\Delta$, and $\alpha\in (0, 1)$, \[\text{If }w\,A\,q\sim p\text{ and }w A\,q\sim_A \alpha\,p+(1-\alpha)w,\text{ then }\] 
\[w\,B\,q\sim r\text{ implies }w\,B\,q\sim_B \alpha\,r+(1-\alpha)w.\] 
\end{axm}

\begin{thm}\label{WIEU} Suppose $\succsim$ has a full-support. The following are equivalent.

\begin{itemize} 
\item[(i)] A family of preference relations $\{\succsim_E\}_{E\in\Sigma}$ admits a \textbf{\nameref{WIEUdef} representation}. 
\item[(ii)] It satisfies \nameref{SEU}, \nameref{CDCoh}, \nameref{PC}, and \nameref{RTC}. 
\item[(iii)] It admits a \textbf{\nameref{WIEUdef} representation with respect to a continuous, strictly convex distance function}. 
\end{itemize}
\end{thm}

\section{Related Literature}\label{literature}

A few papers have studied the idea of distance minimization and how it relates to belief updating.   \cite{Perea2009} axiomatized \emph{imaging} rules, which are minimum distance rules utilizing Euclidean distance. Under imaging, for each $E \subseteq S$ a posterior $\pi$ is selected that minimizes $d_\mu(\pi)=\|\phi(\mu)-\phi(\pi)\|$, where $\pi \in \Delta(E)$ and $\phi$ is an affine function. This is a special case of the \nameref{IEU}. More recently, \cite{Basu_2018} studies AGM \citep{Alchourron1985} belief revision. Within this setting, he establishes an equivalence between lexicographic updating rules and updating rules that are AGM-consistent, Bayesian, and weak path independent. He then turns to  minimum distance updating rules and shows that every support-dependent lexicographic updating rule admits a minimum distance representation. In contrast, we allow for non-Bayesian updating. \cite{zhao2020} and \cite{DKTgeneral} both study distance minimization ``general information;'' information is a subset $I$ of $\Delta(S)$ rather than an event. This more general notion of information requires significantly different axioms. Moreover, \cite{zhao2020} focuses on Bayes' rule. 
 
 There is a large literature in experimental economics and psychology documenting various belief biases, and excellent surveys can be found in \cite{Camerer1995} and \cite{benjamin2019}. There is also growing number of papers taking axiomatic approaches to studying forms of non-Bayesian updating.\footnote{For behavioral models of non-Bayesian updating, see, for example \cite{barberis1998model}; \cite{rabin1999}; \cite{mullainathan2002memory}; \cite{rabin2002inference}; \cite{mullainathan2008coarse}; \cite{gennaioli2010comes}; and \cite{bordalo2016stereotypes}.}  Of course, \cite{ortoleva2012} is the most closely related among these, and has already been discussed in detail. Other papers include \cite{Suleymanov2021}, which studies deviations from Bayesian updating caused by ambiguity; \cite{jakobsen2022}, which studies a ``nearly Bayesian'' updater that selects between subjectively plausible posteriors; \cite{Epstein2006}  and \cite{kovach2020}, both of which study a prior-biased updating rule in which posterior beliefs are a convex combination of the prior and the Bayesian posterior; and \cite{Epstein2008}, which extends \cite{Epstein2006} to an infinite horizon setting. \cite{KWZ2021} studies a rule that also involves a convex combination between prior beliefs and a ``recommended belief,'' but does so in the context of general information (i.e., subsets of $\Delta(S))$ so it is not directly comparable. The updating rule in \cite{Epstein2006}, \cite{kovach2020},  and \cite{Epstein2008} is a special case of Weighted \nameref{IEU} characterized in \autoref{sect-wiu}. 

Our paper also contributes to a growing literature applying models of non-standard belief updating rules to games of strategic information transmission. Recent contributions in this are include \cite{galperti2019persuasion}, \cite{de2022non}, and \cite{lee2023cheap}. 

As we carefully discussed in \autoref{zero}, updating under zero-probability events is studied in \cite{Myerson1886a, Myerson1886b} and \cite{ortoleva2012}. Another well-known approach to dealing with null events is the Lexicographic Probability System (LPS) of \cite{Blume1991}.  While LPS also involves a collection of probability distributions, LPS utilizes the entire collection of distributions in the evaluation process via a lexicographic ordering. Consequently, a DM described by LPS will violate Archimedean Continuity, (see \autoref{SEU}(ii)) of the initial preference. Further, LPS replaces (Savage) null-events with ``infinitely more likely than,'' so that null-events are effectively precluded. While LPS necessarily deviates from SEU, there is a mathematical equivalence between conditional probabilities generated by LPS and CPS (e.g., see \cite{brandenburger2006notes}). Hence, our results further clarify the connections between HT, CPS, and LPS.

\appendix

\section{Proofs}

\subsection{Proof of \autoref{prp:BD}}

Take any non-null $E\in\Sigma$. Let $\text{sp}(\mu)=A\cup C$ and $E=B\cup C$ where $\text{sp}(\mu)\cap E=C$.    We then solve the following optimization problem:
\[\max_{\pi\in\Delta(E)} \sum^n_{i=1}\mu_i\, \sigma\left(\frac{\pi_i}{\mu_i}\right)=\sum_{i\in A\cup C}\mu_i\, \sigma\left(\frac{\pi_i}{\mu_i}\right)=\sum_{i\in C}\mu_i\, \sigma\left(\frac{\pi_i}{\mu_i}\right)+\mu(A)\, \sigma(0).\]
Hence we want to maximize $f((\pi_i)_{i\in C})=\sum_{i\in C}\mu_i\, \sigma\left(\frac{\pi_i}{\mu_i}\right)$ subject to the constraint $\sum_{i\in C} \pi_i=1-\pi(A)$. Let us first fix $\alpha=1-\pi(A)$ and $C'=\{i\in C|\pi_i>0\}$. Then we need to maximize 
\[\sum_{i\in C'}\mu_i\, \sigma\left(\frac{\pi_i}{\mu_i}\right)-\lambda(\sum_{i\in C'} \pi_i-\alpha).\]
The first order condition gives $\sigma'(\frac{\pi_i}{\mu_i})=\lambda$ for each $i\in C'$ (Since $\sigma$ is strictly concave, the FOC is sufficient). Hence, $\pi_i=\mu_i\,c'^{-1}(\lambda)$. After finding $\lambda$ from the constraint $\sum_{i\in C'} \pi_i=\alpha$, we have $\pi_i=\alpha\,\frac{\mu_i}{\mu(C')}$. If we calculate the objective function at the above values:
\[f((\pi_i)_{i\in C})=\mu(C')\, \sigma\left(\frac{\alpha}{\mu(C')}\right)+\mu(C\setminus C')\sigma(0)\]
We need to find the optimal $\alpha$ and $C'$. Let us prove that $\mu(C)>\mu(C')$ implies
\[\mu(C)\, \sigma\left(\frac{\alpha}{\mu(C)}\right)>\mu(C')\, \sigma\left(\frac{\alpha}{\mu(C')}\right)+\mu(C\setminus C')\sigma(0);\]
equivalently,
\[\mu(C)\, \big(\sigma\left(\frac{\alpha}{\mu(C)}\right)-\sigma(0)\big)>\mu(C')\,\big(\sigma\left(\frac{\alpha}{\mu(C')}\right)-\sigma(0)).\]
To obtain the above inequality, it is sufficient to show that $x(\sigma(\frac{\alpha}{x})-\sigma(0))$ is strictly increasing; i.e., $(x(\sigma(\frac{\alpha}{x})-\sigma(0)))'=\sigma(\frac{\alpha}{x})-\sigma(0)-\frac{\alpha}{x}\,\sigma'(\frac{\alpha}{x})>0$. The inequality $\sigma(\frac{\alpha}{x})-\sigma(0)>\frac{\alpha}{x}\,\sigma'(\frac{\alpha}{x})$ holds since $\sigma$ is strictly concave. Hence, $f$ is maximized when $C'=C$.

Since $\sigma$ is strictly increasing, we also have $\mu(C)\, \sigma\left(\frac{1}{\mu(C)}\right)>\mu(C)\, \sigma\left(\frac{\alpha}{\mu(C)}\right)$ when $1>\alpha$. Hence, $f$ is maximized when $\alpha=1$ and $C'=C$. In other words, $\pi_i=\frac{\mu_i}{\mu(C)}$; i.e., $\mu_E=\text{BU}(\mu, E)$.

\subsection{Lemma 1}

The following result will be useful.

\begin{lem} For any $\mu, \pi\in\Delta(S)$, $-\sigma(0)\ge \beta^\sigma(\mu, \pi) \ge -\sigma(1)$. 
\end{lem}

\begin{proof}[Proof of Lemma 1]
Since $\sigma$ is strictly increasing, it is immediate that $\beta^\sigma(\mu, \pi)\le -\sigma(0)$. For any $C\in\Sigma$, let 
\[f(C)=\mu(C)\,\sigma\big(\frac{1}{\mu(C)}\big)+(1-\mu(C))\,\sigma(0).\]
As we showed in the proof of Proposition 1, $x\,\sigma(\frac{1}{x})+(1-x)\,\sigma(0)$ is strictly increasing when $x\in (0, 1)$. Hence we have, $\sigma(1)\ge f(C)$. Let $A=\text{sp}(\mu)\bigcap \text{sp}(\pi)$. By Proposition 1, $f(A)\ge -\beta^\sigma(\mu, \pi)$. Hence, $\beta^\sigma(\mu, \pi)\ge -\sigma(1)$.
\end{proof}



\subsection{Proof of Proposition \ref{os2}}

We first consider the scenario where $E$ is a null-event. Then for any $\pi\in\Delta(E)$, we have 
\[d_\mu(\pi)=\beta^\sigma(\mu^*, \pi).\]
Then by Proposition 1, we have $\mu_E=\text{BU}(\mu^*, E)$. Suppose now $E$ is non-null. Let $\text{sp}(\mu)=A\cup C$ and $E=B\cup C$ where $\text{sp}(\mu)\cap E=C$. 

\[d_\mu(\pi)=
\begin{cases}
\beta^\sigma(\mu, \pi)&\text{ if }\pi\in\Delta(E)\setminus \Delta(B),\\
\beta^\sigma(\mu^*, \pi)+\sigma(1)+|\sigma(0)|&\text{ if }\pi\in\Delta(B).
\end{cases}\]
Let $\mu^1=\text{BU}(\mu, E)$ and $\mu^2=\text{BU}(\mu, B)$. 
By Proposition 1, $\mu^1$ maximizes $\beta^\sigma(\mu, \pi)$ subject to the constraint $\pi\in\Delta(E)$. Again, by Proposition 1, $\mu^2$ maximizes $\beta^\sigma(\mu^*, \pi)$ subject to the constraint $\pi\in\Delta(B)$. Hence, to show that $\mu_E=\mu^1$, it is sufficient to prove that $d_\mu(\mu^1)<d_\mu(\mu^2)$; equivalently,
\[d_\mu(\mu^1)=\beta^\sigma(\mu, \mu^1)<d_\mu(\mu^2)=\beta^\sigma(\mu^*, \mu^2)+\sigma(1)+|\sigma(0)|.\]
The above inequality is implied by Lemma 1.

\subsection{Proof of \autoref{repthrm}}

Note that (iii) trivially implies (i). Let us first show that (i) implies (ii). Suppose $\{\succsim_E\}$ admits an \nameref{IEU} representation with respect to $(\mu, u, d_\mu)$. The \nameref{IEU} representation indeed satisfies \nameref{SEU}. We now prove the necessity of Consequentialism and Dynamic Coherence.

\medskip
\noindent\textbf{Consequentialism.} Take any $E\in\Sigma$ and $f, g\in F$ such that $f(s)=g(s)$ for all $s\in E$. Since $\mu_E(E)=1$ and $f(s)=g(s)$ for all $s\in E$, we have  
\[\sum_{s\in S} \mu_E(s)f(s)=\sum_{s\in E} \mu_E(s)f(s)=\sum_{s\in S} \mu_E(s)g(s)=\sum_{s\in E} \mu_E(s)g(s);\] 
i.e., $f\sim_{E} g$.

\medskip
\noindent\textbf{Dynamic Coherence.} Take any $A_1, \ldots, A_n\subseteq S$ such that $S\setminus A_i$ is $\succsim_{A_{i+1}}$-null for each $i\le n-1$ and $S\setminus A_n$ is $\succsim_{A_1}$-null. Equivalently, $\mu_{A_{i+1}}(A_i)=1$ for each $i\le n-1$ and $\mu_{A_1}(A_n)=1$. Since $\mu_{A_{i+1}}\in\Delta(A_i)$ and $\mu_{A_i}=\arg\min_{\pi\in\Delta(A_i)}d_\mu(\pi)$, $d_\mu(\mu_{A_i})\le d_\mu(\mu_{A_{i+1}})$. Similarly, we have $d_\mu(\mu_{A_n})\le d_\mu(\mu_{A_{1}})$. Therefore, we have 
\[d_\mu(\mu_{A_1})\le d_\mu(\mu_{A_{2}})\le \ldots \le d_\mu(\mu_{A_{n}})\le d_\mu(\mu_{A_{1}});\]
i.e., $d_\mu(\mu_{A_1})=d_\mu(\mu_{A_{n}})$. Since $\mu_{A_n}$ is the unique minimizer of $d_\mu$ in $\Delta(A_n)$ and $\mu_{A_1}\in\Delta(A_n)$, $d_\mu(\mu_{A_1})=d_\mu(\mu_{A_{n}})$ implies that $\mu_{A_1}=\mu_{A_{n}}$; i.e., $\succsim_{A_1}=\succsim_{A_{n}}$. 

\bigskip
Let us now show that (ii) implies (iii). Suppose $\{\succsim_E\}_{E\in\Sigma}$ satisfies \nameref{SEU}, \nameref{CON}, and \nameref{DCoh}. Since $\succsim$ satisfies SEU postulates, there is $(\mu, u)$ such that $\succsim$ has a SEU representation with  $(\mu, u)$. Since $\succsim_E$ satisfies SEU postulates, there is $(\mu_E, u_E)$ such that $\succsim_E$ has a SEU representation with  $(\mu_E, u_E)$. By Invariant Risk Preference, $u_E(p)\ge u_E(q)$ and $u(p)\ge u(q)$ for any $p, q\in\Delta(X)$. Without loss of generality, let us assume that $u_E=u$. Hence, $\succsim_E$ has a SEU representation with  $(\mu_E, u)$.

Let us now discuss the implications of \nameref{CON}. Take any $E \in \Sigma$ and any $f,g\in \F$ and $p, q\in\Delta(X)$ such that $p\succ q$ and $f(s) =g(s)=p$ for all $s \in E$ and $f(s)=p$ and $g(s)=q$ for any $s\in E^c$. By \nameref{CON}, we have $f\sim_{E} g$; equivalently, 
\[\sum_{s\in S} \mu_E(s)f(s)=u(p)=\sum_{s\in E} \mu_E(s)g(s)=\mu_E(E)u(p)+(1-\mu_E(E))u(q).\] 
In other words, we have $\mu_E(E)=1$; i.e.,  $\mu_E\in \Delta(E)$.

\bigskip
\noindent\textbf{Afriat's theorem for general budget sets.} To obtain the IU representation, we use an extension of Afriat's theorem (\cite{afriat1967construction}) for general budget sets due to \cite{matzkin1991axioms}. To state Afriat's theorem for general budget sets, some notation is necessary. Let $Z$ be a convex, bounded subset of $\mathbb{R}^n_{+}$. Let $\mathscr{D}=(\mathbf{x}^t, B^t)_{t\in T}$ be a data set where $\mathbf{x}^t\in B^t$ is the observed consumption bundle that is chosen from the budget set $B^t\subset Z$ at observation $t\in T$. We say that $(\mathbf{x}^t, B^t)$ is a co-convex subset of $Z$ if the following three conditions hold: (i) $Z\setminus B^t$ is open and convex; (ii) for any $\mathbf{e}\ge 0$ and $\mathbf{x}\in Z\setminus B^t$, $\mathbf{x}+\mathbf{e}\in Z$ implies $\mathbf{x}+\mathbf{e}\in Z\setminus B^t$; and (iii) for any $\mathbf{e}>0$, $\mathbf{x}^t+\mathbf{e}\in Z$ implies $\mathbf{x}^t+\mathbf{e}\in Z\setminus B^t$.

Let us now define the following revealed preference relation on $\{\mathbf{x}^t\}_{t\in T}$. We say $\mathbf{x}^t$ is revealed preferred to $\mathbf{x}^s$, denoted by $\mathbf{x}^t\succsim_R \mathbf{x}^s$ if $\mathbf{x}^s\in B^t$. We say $\mathbf{x}^t$ is strictly revealed preferred to $\mathbf{x}^s$, denoted by $\mathbf{x}^t\succ_R \mathbf{x}^s$ if $\mathbf{x}^s\in B^t$ and $\mathbf{x}^t\neq \mathbf{x}^s$. Finally, we say the data set $\mathscr{D}=(\mathbf{x}^t, B^t)_{t\in T}$ satisfies the Strong Axiom of Revealed Preferences (SARP) if $\succsim_R$ is acyclic; i.e., there is no sequence $\mathbf{x}^{t_1}, \mathbf{x}^{t_2}, \ldots, \mathbf{x}^{t_L}$ such that $\mathbf{x}^{t_l}\succsim_R \mathbf{x}^{t_{l+1}}$ for each $l\le L-1$ and $\mathbf{x}^{t_L}\succ_R\mathbf{x}^{t_1}$.

\medskip
\noindent\textbf{Theorem 1 of Matzkin (1991).} Suppose for each $t\in T$, $(\mathbf{x}^t, B^t)$ is a co-convex subset of $Z$. Then the data set $\mathscr{D}=(\mathbf{x}^t, g^t)_{t\in T}$ satisfies SARP if and only if there is a strictly increasing, continuous, strictly concave utility function $u:Z\to\mathbb{R}$ such that for any $t\in T$,
\[u(\mathbf{x}^t)>u(\mathbf{x})\text{ for any }\mathbf{x}\in B^t\setminus\{\mathbf{x}^t\}.\]

To apply the above theorem, let us arbitrarily label the set of all events: $\Sigma=\{E_t\}_{t\in T}$. Then let $Z=\Delta(S)$ and $\mathbf{x}^t=\mu_{E_t}$ and $B^t=\Delta(E_t)$ for each $t\in T$. Let $\mathscr{D}=(\mathbf{x}^t, B^t)_{t\in T}$. 

Note that $Z$ is a convex, bounded subset of $\mathbb{R}^n_{+}$. Let us show that $(\mathbf{x}^t, B^t)$ is a co-convex subset of $Z$. First, $Z\setminus B^t$ is open and convex in $Z$. Second, for any $\mathbf{x}\in Z$ and $\mathbf{e}\ge 0$, $\mathbf{x}+\mathbf{e}\in Z$ implies $\mathbf{e}=0$. Hence, $(ii)$ and $(iii)$ of co-convexity are trivially satisfied.

Let us now show that \nameref{DCoh} implies that $\mathscr{D}=(\mathbf{x}^t, B^t)_{t\in T}$ satisfies SARP. Take any sequence $\mathbf{x}^{t_1}, \mathbf{x}^{t_2}, \ldots, \mathbf{x}^{t_L}$ such that $\mathbf{x}^{t_l}\succsim_R \mathbf{x}^{t_{l+1}}$ for each $l\le L-1$ and $x^{t_L}\succsim_R\mathbf{x}^{t_1}$. To prove SARP, we shall show that $\mathbf{x}^{t_L}=\mathbf{x}^{t_1}$. By definition of the revealed preference relation $\succsim_R$, $\mathbf{x}^{t_l}\succsim_R \mathbf{x}^{t_{l+1}}$ is equivalent to $\mathbf{x}^{t_{l+1}}\in \Delta(E^{t_l})$. In other words, $\mu_{E_{t_{l+1}}}\in\Delta(E_{t_l})$ for each $l\le L-1$. Similarly, $\mu_{E_{t_1}}\in\Delta(E_{t_L})$. 

Note that $\mu_{E_{t_{l+1}}}\in\Delta(E_{t_l})$ implies $\mu_{E_{t_{l+1}}}(E_{t_l})=1$; equivalently, $\mu_{E_{t_{l+1}}}(S\setminus E_{t_l})=0$. In other words, $S\setminus E_{t_l}$ is $\succsim_{E_{t_{l+1}}}$-null for each $l\le L-1$. Similarly, $S\setminus E_{t_L}$ is $\succsim_{E_{t_{1}}}$-null. By \nameref{DCoh}, $\succsim_{E_{t_1}}=\succsim_{E_{t_L}}$; equivalently, $\mu_{E_{t_1}}=\mu_{E_{t_L}}$. In other words, $\mathbf{x}^{t_1}=\mathbf{x}^{t_L}$.

Since $\mathscr{D}=(\mathbf{x}^t, B^t)_{t\in T}$ satisfies SARP, by Theorem 1 of Matzkin (1991), 
there is a strictly increasing, continuous, strictly concave utility function $u:Z\to\mathbb{R}$ such that for any $t\in T$,
\[u(\mathbf{x}^t)>u(\mathbf{x})\text{ for any }\mathbf{x}\in B^t\setminus\{\mathbf{x}^t\}.\]
Let $d_\mu=-u$. Then since $B^t=\Delta(E_t)$ and $\mathbf{x}^t=\mu_{E_t}$, 
\[\mu_{E_t}=\arg\min_{\pi\in \Delta(E_t)} d_\mu(\pi).\]
Finally, note that $d_\mu$ is continuous and strictly convex.

\subsection{Proof of \autoref{prp:DC}}

This follows directly from existing results on Dynamic Consistency. For example, see \cite{Ghirardato2002}.

\subsection{Proof of \autoref{prp:d-BD}}


We start by constructing a distortion $\delta_E$ for an arbitrary non-null event $E$. Without loss, suppose $|E| \ge 2$. Fix $s^* \in E$ with $\mu_E(s^*)>0$. For all $s \in E$, let $\delta_E(\mu(s))=\frac{\mu_E(s)}{\mu_E(s^*)}$. Consider any $s_1,s_2 \in E$ such that $\mu(s_1)=\mu(s_2)$.  Then by \nameref{Consistency} it follows that $\mu_E(s_1)=\mu_E(s_2)$, and by construction of $\delta_E$, it follows that $\delta_E(\mu(s_1))=\frac{\mu_E(s_1)}{\mu_E(s^*)}= \frac{\mu_E(s_2)}{\mu_E(s^*)}=\delta_E(\mu(s_2))$, hence $\delta_E$ is well-defined. Finally, note that for any $s,s' \in E$, $\frac{\delta_E(\mu(s'))}{\delta_E(\mu(s))} =  \frac{\mu_E(s')}{\mu_E(s)}$. Summing over $s' \in E$ and using $\sum_{s' \in E}\mu_E(s')=1$ yields $\mu_E(s)=\frac{\delta_E(\mu(s))}{\sum_{s' \in E}\delta_E(\mu(s'))}$, hence $\mu_E=BU(\delta_E(\mu), E)$.

Next, we use \nameref{IIA} to show that $\delta_E$ is in fact independent of $E$. Fix any $E_1,E_2$ with $s,s' \in E_1 \cap E_2$. Consider some $p,q$ such that $p\{s\}r \sim_{E_1} q\{s'\}r$. It is without loss to suppose $u(r)=0$, and hence $u(p)\mu_{E_1}(s) = u(q) \mu_{E_1}(s')$. By the previous result, it follows that \[u(p)\frac{\delta_{E_1}(\mu(s))}{\sum_{\tilde{s} \in E_1}\delta_{E_1}(\mu(\tilde{s}))} = u(q) \frac{\delta_{E_1}(\mu(s'))}{\sum_{\tilde{s} \in E_1}\delta_{E_1}(\mu(\tilde{s}))},\]
and so \[\frac{\delta_{E_1}(\mu(s))}{\delta_{E_1}(\mu(s'))}=\frac{u(q)}{u(p)}.\]
By applying \nameref{IIA}, it follows that $p\{s\}r \sim_{E_2} q\{s'\}r$, and hence 
\[\frac{\delta_{E_1}(\mu(s))}{\delta_{E_1}(\mu(s'))}=\frac{u(q)}{u(p)} = \frac{\delta_{E_2}(\mu(s))}{\delta_{E_2}(\mu(s'))}.\]

Hence there exists a $\delta:[0, 1]\to \mathbb{R}_{+}$ such that for any non-null $E$, $\mu_E=BU(\delta(\mu), E)$. Finally, since $\delta$ is clearly only unique up to a scalar, it is without loss to suppose that $\delta:[0, 1]\to [0,1]$.\footnote{It is clear from our proof that \nameref{Consistency} and \nameref{IIA} can be imposed for null-events and we obtain Distorted Bayesian with $\delta>0$.}

\subsection{Proof of \autoref{prp:mon}}

It is clear that \nameref{Mon} implies \nameref{Consistency}, and by the previous result we have some $\delta:[0, 1]\to [0, 1]$ such that $\mu_E=BU(\delta(\mu), E)$ for any non-null $E$. Consider any $E$ and $s,s' \in E$ and suppose $\mu(s) > \mu(s)$. Then from \nameref{Mon}, if $x \succ y$ if follows that $x \{s\} y\succ x \{s'\} y$ and thus $x \{s\} y\succ_{E} x \{s'\} y$, which implies $\mu_E(s) > \mu_E(s')$. From here it is immediate that $\delta(\mu(s)) > \delta(\mu(s'))$. Since $\delta$ is arbitrary outside of $\{\mu(s)\}_{s \in S}$, it can extended to $[0,1]$ so that $\delta$ is strictly increasing.

\subsection{Proof of \autoref{CPSthrm}}

Necessity of the axioms is trivial, so we only prove sufficiency. By \nameref{SEU}, there are $U$ and $\{\mu_E\}_{E\in\Sigma}$ such that for any $E\in\Sigma$, $\succsim_{E}$ admits a SEU representation with $(\mu_E, u)$; for all $f,g \in \cal F$:
\[f \succsim_{E} g \quad \text{if and only if} \quad \sum_{s\in E}  U\big(f(s)\big)\mu_E(s) ~ \geq~ \sum_{\omega \in S_E}  U\big(g(s)\big) \mu_E(s).\]
By \nameref{CON}, $\mu_E(E)=1$. We now shall show Equation \ref{CPSeq}. Take any $E\in\Sigma$ and $A\subseteq E$. Let $A \in \Sigma$ be a $\succsim_{E}$-non-null event; i.e., $\mu_E(A)>0$. Consider acts $f_Ah$ and $g_Ah$ such that  $f_Ah \succsim_{E} g_Ah$; i.e.,
\begin{eqnarray}
\sum_{s \in S}  u\big(f_Ah(s)\big)\mu_E(s) ~ \geq~ \sum_{s \in S}  u\big(g_Ah(s)\big) \mu_E(s).
\end{eqnarray}
By \nameref{CON}, 
\begin{eqnarray}
\sum_{s \in E}  u\big(f_Ah(s)\big)\mu_E(s) ~ \geq~ \sum_{s \in E}  u\big(g_Ah(s)\big) \mu_E(s),
\end{eqnarray}
or equivalently, 
\begin{eqnarray}\label{EQ-1-BU}
\sum_{s \in A}  u\big(f(s)\big)\mu_E(s) ~ \geq~ \sum_{s \in A}  u\big(g(s)\big) \mu_E(s).
\end{eqnarray}
By \nameref{CC}, $f\succsim_{A} g$; i.e.,
\begin{eqnarray}\label{EQ-2-BU}
\sum_{s \in A}  u\big(f(s)\big)\mu_A(s) ~ \geq~ \sum_{s \in A}  u\big(g(s)\big) \mu_A(s).
\end{eqnarray}
Since Equations \eqref{EQ-1-BU} hold for any $f, g$, we have Bayesian updating
\begin{eqnarray}
\mu_A(s)=\frac{\mu_E(s)}{\mu_E(A)}~\text{for each}~s\in A.
\end{eqnarray}

Finally, if $\mu_E(A)=0$, then we have $\mu_E(s)=0$ for any $s\in A$. Hence, Equation \ref{CPSeq} holds. 

\subsection{Proof of \autoref{CPS=MD}}

Suppose $\{\succsim_{E}\}_{E\in \Sigma}$ admits a CPS representation. To prove the first part of this proposition, we construct $\mu^0, \ldots, \mu^K\in\Delta(S)$ inductively. Let $S_0=S$, and consider $\succsim_{S_0}$. We let $\mu^0=\mu_{S_0}$, and if $\mu^0$ has full support, stop. Otherwise, let $S_1$ denote the set of all $\succsim_{S_0}$-null states, and let $\mu^1=\mu_{S_1}$. By \nameref{CON}, $\mu^1(S_0\setminus S_1)=0$. If $\text{sp}(\mu^1)=S_1$, stop. Otherwise, let $S_2$ denote the set of all $\succsim_{S_1}$-null states, and $\mu^2=\mu_{S_2}$. We proceed in this fashion until we reach a $K$ such that $\text{sp}(\mu^K)=S_K$. Since $S$ is finite, we must eventually stop. Note that we have constructed $\mu^0, \ldots, \mu^K\in\Delta(S)$ such that $\text{sp}(\mu^0), \ldots, \text{sp}(\mu^K)$ is a partition of $S$. We now shall prove that for any $E\in\Sigma$, $\mu_{E}=\text{BU}(\mu^{k^*}, E)$ where $k^*=\min\{k \mid \mu^k(E) > 0\}$.

Since $k^*=\min\{k \mid \mu^k(E) > 0\}$, $E\subseteq S_{k^*}=\bigcup_{k\ge k^*} \text{sp}(\mu^k)$. By the construction, $\mu^{k^*}=\mu_{S_{k^*}}$. Hence, $\mu_{S_{k^*}}(E)>0$. Then by Equation \ref{CPSeq}, for any $s\in E$, $\mu_E(s)=\frac{\mu^{k^*}(s)}{\mu^{k^*}(E)}$;
equivalently, $\mu_{E}=\text{BU}(\mu^{k^*}, E)$.

We now shall show that $\{\succsim_{E}\}_{E\in \Sigma}$ has an \nameref{IEU} representation with respect to the following distance function: 
\[d_\mu(\pi)=\beta^\sigma(\mu^{k^*}, \pi)+k^*\,\big(\sigma(1)+|\sigma(0)|\big),\]
where $k^*=\min\{k \mid \mu^k( \text{sp}(\pi))>0\}$. It is enough to show that for any $E\in\Sigma$, 
\[\mu_E=\argmin_{\pi\in \Delta(E)} d_\mu(\pi).\]
Take any $E$ and let $k^*=\min\{k \mid \mu^k(E) > 0\}.$ Note that for any $\pi\in\Delta(E)$, $\min\{k \mid \mu^k( \text{sp}(\pi))>0\}\ge k^*$. Hence, $\Delta_{k^*}, \ldots, \Delta_{K}$ be the partition of $\Delta(E)$ such that for any $\pi\in\Delta(E)$, $\pi\in \Delta_l$ if and only if $l=\min\{k \mid \mu^k( \text{sp}(\pi))>0\}$. Let $\rho^l=\arg\min_{\pi\in\Delta_l} d_{\mu}(\pi)$. By Proposition 1, $\rho^{k^*}=BU(\mu^{k^*}, E)$. Take any $l>k^*$. We shall show $d_\mu(\rho^{k^*})<d_\mu(\rho^{l})$; equivalently,
\[\beta^\sigma(\mu^{k^*}, \rho^{k^*})-\beta^\sigma(\mu^{l}, \rho^{l})<(l-k^*)\big(\sigma(1)+|\sigma(0)|\big).\] 
The above inequality is implied by Lemma 1.

\subsection{Proof of \autoref{CPS-HTE}}

See the proof of \autoref{eCPS} as this corollary is a special case of \autoref{eCPS}  when $\epsilon=0$. Alternatively, Proposition 3 and Theorem 2 also imply this corollary.

\subsection{Proof of \autoref{eCPS}}

Let $\{\succsim_E\}$ be a family of preference relations with an $\epsilon$-CPS representation for some $\epsilon\in [0, 1)$. Then, there are probability distributions $\mu_0, \ldots, \mu_K$ such that 
\[\mu_E=\text{BU}(\mu_{k^*}, E)\text{ where }k^*=\min\{k\le K \mid \mu_k(E)>\epsilon\}\]
 for every $E\in\Sigma$. Let $\Sigma_0, \ldots, \Sigma_K$ be a partition of $\Sigma$ such that for each $k$, $\Sigma_k$ is the collection of events for which the prior $\mu_k$ is used for updating: \[\Sigma_k=\{E\in \Sigma \mid k=\min\{\tilde{k}\le K \mid \mu_{\tilde{k}}(E)>\epsilon\}\}.\] Throughout this proof, we assume that for any $k\le K$, $E_k$ is an element of $\Sigma_k$. Take $\overline{\rho}_0, \underline{\rho}_0, \ldots, \overline{\rho}_K, \underline{\rho}_K$ with \[\overline{\rho}_0>\underline{\rho}_0>\overline{\rho}_1>\underline{\rho}_1>\ldots>\overline{\rho}_K>\underline{\rho}_K>\delta\,\overline{\rho}_0>0\]
 and $\underline{\rho}_k>\overline{\rho}_k\,\mu^{E'_k}(E_k)$ for any $E_k, E'_k$ with $\mu^{E'_k}(E_k)<1$. 
 
 Let $\mu^E_k=\text{BU}(\mu_k, E)$ for any $E\in \Sigma$. Let $\rho$ be an element of $\Delta(\{\mu^{E_k}_k\}_{k\le K, E_k\in \Sigma_k})$ such that (i) $\rho(\mu ^{E_k}_k)\in (\underline{\rho}_k, \overline{\rho}_k)$ for any $k\le K$ and (ii) $\rho(\mu ^{E_k}_k)>\rho(\mu ^{E'_k}_k)$ if $\mu ^{E_k}_k\neq \mu ^{E'_k}_k$ and $\mu ^{E'_k}_k(E_k)=1$. 
 
 Let us first show that there is $\rho$ that satisfies (ii). Let $\mu ^{E_k}_k\succ^* \mu ^{E'_k}_k$ if $\mu ^{E_k}_k\neq \mu ^{E'_k}_k$ and $\mu ^{E'_k}_k(E_k)=1$. It is enough to show that $\succ^*$ is acyclic. To show acyclicity, suppose that there are $E^1_k, \ldots, E^T_k$ such that $\mu ^{E^t_k}_k(E^{t+1}_k)=1$ for each $t\le T-1$ and $\mu ^{E^T_k}_k(E^1_k)=1$. Note that  $\mu ^{E'_k}_k(E_k)=1$ is equivalent to $\text{sp}(\mu_k)\cap E'_k\subseteq E_k$. Hence, $\mu ^{E'_k}_k(E_k)=1$ implies $\text{sp}(\mu_k)\cap E'_k\subseteq \text{sp}(\mu_k)\cap E_k$. Then, $\mu ^{E^t_k}_k(E^{t+1}_k)=1$ implies $\text{sp}(\mu_k)\cap E^{t}_k\subseteq \text{sp}(\mu_k)\cap E^{t+1}_k$ and $\mu ^{E^T_k}_k(E^1_k)=1$ implies  $\text{sp}(\mu_k)\cap E^T_k\subseteq \text{sp}(\mu_k)\cap E^1_k$. Hence,  $\text{sp}(\mu_k)\cap E^t_k=\text{sp}(\mu_k)\cap E^{t'}_k$ for any $t, t'$; i.e., $\mu ^{E^t_k}_k=\mu ^{E^{t'}_k}_k$.    
 
 We now show that $\{\succsim_E\}$ has a HT representation with $(\rho, \delta)$ when $\delta$ is large enough. Hence, we shall show that for any $E_k$, 
 
 \[\rho(\mu^{E_k}_k)\mu^{E_k}_k(E_k)=\rho(\mu^{E_k}_k)>\rho(\mu^{E_j}_j)\mu^{E_j}_j(E_k)\text{ for any }j\neq k.\]
For any $j>k$, the above holds since $\rho(\mu^{E_k}_k)>\rho(\mu^{E_j}_j)$. Suppose now $j<k$. In this case, $\mu_j(E_k)\le \epsilon$ since $k$ is the lowest index such that $\mu_k(E_k)>\epsilon$. Then, $\mu^{E_j}_j(E_k)=\text{BU}(\mu_j, E_j)(E_k)=\frac{\mu_j(E_k\cap E_j)}{\mu_j(E_j)}$. Since $\mu_j(E_k)\le \epsilon$ and $\mu_j(E_j)>\epsilon$, there is a large enough $\delta\in [0, 1)$ such that $\mu^{E_j}_j(E_k)\le \delta$. Hence, by the construction of $\rho$, 
 \[\rho(\mu^{E_k}_k)>\delta \rho(\mu^{E_j}_j)\ge \rho(\mu^{E_j}_j)\mu^{E_j}_j(E_k).\]
We finally show that the HT representation correctly chooses $\mu^{E_k}_k$ among $\{\mu^{E'_k}_k\}_{E'_k}$ for each $E_k$. When $\mu^{E'_k}_k(E_k)<1$, we have 
 \[\rho(\mu^{E_k}_k)\mu^{E_k}_k(E_k)=\rho(\mu^{E_k}_k)>\underline{\rho}_k>\overline{\rho}_k\mu^{E'_k}_j(E_k)>\rho(\mu^{E'_k}_k)\mu^{E'_k}_j(E_k).\]
When $\mu^{E'_k}_k(E_k)=1$ and $\mu^{E_k}_k\neq \mu^{E'_k}_k$, \[\rho(\mu^{E_k}_k)\mu^{E_k}_k(E_k)=\rho(\mu^{E_k}_k)>\rho(\mu^{E'_k}_k)=\rho(\mu^{E'_k}_k)\mu^{E'_k}_j(E_k).\]

It is immediate from the above construction of $\delta$, $\delta=0$ whenever $\epsilon=0$. 

\subsection{Proof of \autoref{concave}}

Let $B=\Omega\setminus A$.  Let $x_i=\pi_{m_1}(\omega_i)$ and $\delta_i=|g\big(\rho_i\big)\, u_i|$. The sender's problem reduces to 
\[\max_{\textbf{x}\in [0, 1]^n} \sum^n_{i=1} \rho_i\, x_i\text{ subject to }\sum^n_{i\in A}\delta_i\,f(x_i)\ge \sum^n_{i\in B}\delta_i\,f(x_i).\]
It is immediate that $x^*_i=1$ whenever $i \in A$. Let $M=\sum_{i\in A}\delta_i\,f(1)$. Then \[\max_{x_i\in [0, 1]} \sum_{i\in B} \rho_i\, x_i\text{ subject to }M\ge \sum_{i\in B}\delta_i\,f(x_i).\] 

\smallskip
\noindent\textbf{Case 1.} $f(x)=x$ and $\frac{g(\rho_i) u_i}{\rho_i}\neq \frac{g(\rho_j) u_j}{\rho_j}$ for any $i, j$ with $u_i, u_j<0$

Note that when $\frac{\rho_i}{\delta_i}>\frac{\rho_j}{\delta_j}$, we cannot have $1>x^*_i$ and $x^*_j>0$. The optimal signal structure takes a form \[x^*_{i_1}=\ldots=x^*_{i_k}=1>x^*_{i_{k+1}}=\frac{M-\sum^k_{s=1} \delta_{i_s}}{\delta_{i_{k+1}}}\ge x^*_{i_{k+2}}=x^*_{i_{|B|}}=0,\] where $\frac{\rho_{i_1}}{\delta_{i_1}}> \ldots > \frac{\rho_{i_k}}{\delta_{i_k}}> \ldots > \frac{\rho_{i_{|B|}}}{\delta_{i_{|B|}}}$. 

\medskip
\noindent\textbf{Case 2.} $f$ is strictly concave.

Let us show that for any $i, j$, we cannot have $x^*_i, x^*_j\in (0, 1)$. Take any $i, j$ and let $\delta_i f(x^*_i)+\delta_j f(x^*_j)=m$. Then $x^*_i, x^*_j$ must be the solution to the following maximization problem 
\[\max_{x_i, x_j\in [0, 1]} \rho_i x_i+\rho_j x_j\text{ subject to }\delta_i f(x_i)+\delta_j f(x_j)=m.\]
From the constraint, we have $x_j=f^{-1}\big(\frac{m-\delta_i f(x_i)}{\delta_j}\big)$. Hence, the above maximization problem reduces to 
\[\max_{x_i\in [a_1, a_2]} \rho_i x_i+\rho_j\, f^{-1}\big(\frac{m-\delta_i f(x_i)}{\delta_j}\big),\]
where $a_1=\max\{0, f^{-1}\big(\frac{m-\delta_j f(1)}{\delta_i}\big)\}$ and $a_2=\min\{1, f^{-1}\big(\frac{m-\delta_j f(0)}{\delta_i}\big)\}$. The objective function is strictly convex since $f$ is strictly concave and $f$ is increasing. Hence, either $x^*_i=a_1$ or $x^*_i=a_2$. Note that $x^*_i=a_1$ means that either $x^*_i=0$ or $x^*_j=1$ and $x^*_i=a_2$ means that either $x^*_i=1$ or $x^*_j=0$. Hence, the optimal signal structure takes a form \[x^*_{i_1}=\ldots=x^*_{i_k}=1>x^*_{i_{k+1}}=\frac{M-\sum^k_{s=1}\delta_{i_s}f(1)-\sum^{|B|}_{s=k+2}\delta_{i_s} f(0)}{\delta_{i_{k+1}}}\ge x^*_{i_{k+2}}=x^*_{i_{|B|}}=0,\]
where $\{i_1, \ldots, i_{|B|}\}$ is a permutation of $B$.

 \subsection{Proof of \autoref{convex}}

Similar to the argument in the proof of \autoref{concave}, we need to solve \[\max_{x_i\in [0, 1]} \sum_{i\in B} \rho_i\, x_i\text{ subject to }M\ge \sum_{i\in B}\delta_i\,f(x_i),\]
where $M=\sum_{i\in A}\delta_i\,f(1)$. As long as $M>0$, there exists $x^*_j>0$. Take any $i\neq j$. Let us show that $x^*_i=0$. Let $\delta_i f(x^*_i)+\delta_j f(x^*_j)=m$. Then $x^*_i, x^*_j$ must be the solution to the following maximization problem 
\[\max_{x_i, x_j\in [0, 1]} \rho_i x_i+\rho_j x_j\text{ subject to }\delta_i f(x_i)+\delta_j f(x_j)=m.\]
From the constraint, we have $x_j=f^{-1}\big(\frac{m-\delta_i f(x_i)}{\delta_j}\big)$. Hence, the above maximization problem reduces to 
\[\max_{x_i\in [a_1, a_2]} \rho_i x_i+\rho_j\, f^{-1}\big(\frac{m-\delta_i f(x_i)}{\delta_j}\big),\]
where $a_1=\max\{0, f^{-1}\big(\frac{m-\delta_j f(1)}{\delta_i}\big)\}$ and $a_2=\min\{1, f^{-1}\big(\frac{m-\delta_j f(0)}{\delta_i}\big)\}$. Since $f'(0)=0$, $(\rho_i x_i+\rho_j\, f^{-1}\big(\frac{m-\delta_i f(x_i)}{\delta_j}\big))'_{x_i}|_{x_i=0}=\rho_i>0$. Hence $x^*_i=0$ cannot be optimal solution. Hence, $x^*_i>0$.

\subsection{Proof of \autoref{WIEU}}

$(ii)\Rightarrow(iii)$. Take any $A\subset S$. Since $\succsim$ and $\succsim_A$ have SEU representations with respect to $(u, \mu)$ and $(u, \mu_A)$, $E$ is $\succsim_A$-unfavored if for any $E'\subseteq S$ and $p, q\in \Delta(X)$, for any $u(p)\mu(E')=u(q)\mu(E)$ implies $u(p)\mu_A(E')\ge u(q)\mu_A(E)$, with at least one strict inequality for some $E'$. Note that when $\mu=\mu_A$, there is no $\succsim_A$-unfavored event since $u(p)\mu(E')=u(q)\mu(E)$ implies $u(p)\mu_A(E')=u(q)\mu_A(E)$ for every $E'$ and $p, q$. However, by Partial Consequentialism, $A^c$ is $\succsim$-unfavored. Hence, $\mu\neq \mu_A$.

If $E$ is $\succsim_A$-unfavored, then $\delta(A)= \frac{\mu_A(E)}{\mu(E)}$ where $\delta(A)=\min_{E'}\frac{\mu_A(E')}{\mu(E')}$. Since $\mu\neq \mu_A$, $\delta(E)<1$. Therefore,
\[E\text{ is }\succsim_A\text{-unfavored iff }\frac{\mu_A(E)}{\mu(E)}=\delta(A).\]

Consider the vector $\mu^*_A=\frac{\mu_A-\delta(A)\mu}{1-\delta(A)}$. For each $s\in S$, since 
$\frac{\mu_A(s)}{\mu(s)}\ge\delta(A)$, $\mu^*_A(s)=\frac{\mu_A(s)-\delta(A)\mu(s)}{1-\delta(A)}\ge 0$. Moreover, $\sum_{s\in S}\mu^*_A(s)=\sum_{s\in S}\frac{\mu_A(s)-\delta(A)\mu(s)}{1-\delta(A)}=1$. Hence, $\mu^*_A\in \Delta(S)$ and \[\mu_A=\delta(A)\,\mu+(1-\delta(A))\mu^*_A.\]
Note that $E$ is $\succsim_A$-unfavored iff $\mu_A(E)=\delta(A)\,\mu(E)$ iff $\mu^*_A(E)=0$. Then by Partial Consequentialism, $A^c$ is $\succsim_A$-unfavored iff $\mu^*_A(A^c)=0$. Hence, $\mu^*_A\in\Delta(A)$.
We now shall show that there is a function $d$ that $\mu^*_A=\arg\min_{\pi\in\Delta(A)} d_\mu(\pi)$. 

We now essentially repeat the part of Theorem 1 for the data set $\mathcal{D}^*=\{(\mu^*_A, \Delta(A))\}_{A\in\Sigma}$ where $\mu^*_S=\mu$. To apply the aforementioned generalization of Afriat's theorem for general budget sets, we first define the following revealed preference relation.
We say that $\mu^*_A$ is strictly revealed preferred to $\mu^*_B$, denoted by $\mu^*_A R^* \mu^*_B$, if $\mu^*_B\in \Delta(A)$ and $\mu^*_A\neq \mu^*_B$. First, note that $\mu^*_S R^* \mu^*_A$ for any $A\in\Sigma\setminus \{S\}$. Second, $\neg \mu^*_A R^* \mu^*_S$ since $\mu\not\in \Delta(A)$. Third, for any $A, B\in\Sigma\setminus \{S\}$, $\mu^*_A R^* \mu^*_B$ implies that $A$ is $\succsim_B$-favored. Hence, Partial Dynamic Coherence is equivalent to the acyclicity of $R^*$.

By the arguments provided in the proof of Theorem 1, $(\mu^*_A, \Delta(A))$ is co-convex. Since $\mathcal{D}^*$ satisfies SARP, by Theorem 1 of Matzkin (1991), there is a strictly increasing, continuous, strictly concave utility function $u:\Delta(S)\to\mathbb{R}$ such that for any $A\in\Sigma$,
\[\mu^*_{A}=\arg\max_{\pi\in \Delta(A)} u(\pi).\]
Let $d_\mu=-u$ and note that $\mu$ is the global minimizer of $d_\mu$ by the previous equation. Moreover, 
\[\mu^*_{A}=\arg\min_{\pi\in \Delta(A)} d_\mu(\pi).\]
Finally, note that $d_\mu$ is continuous and strictly convex. To sum up, we have 
\[\mu_A=\delta(A)\,\mu+(1-\delta(A))\arg\min_{\pi\in \Delta(A)} d_\mu(\pi)\]
for any $A\subseteq S$. We now shall show that $\delta(A)=\delta(B)$. 

Take any $A, B\in\Sigma\setminus S$. There are $p, q, r$ such that $\mu(A^c)u(q)=u(p)$ and $\mu(B^c)u(q)=u(r)$; equivalently, $w\,A\,q\sim p$ and $w\,B\,q\sim_B r$. Since $\mu_A(A^c)=\delta(A)\,\mu(A^c)$, we have $\delta(A)\mu(A^c)u(q)=\mu_A(A^c)u(q)=\delta(A)\, u(p)$; equivalently, $w\,A\,q\sim_A \delta(A)\,p+(1-\delta(A))w$. By \nameref{RTC}, we have $w\,B\,q\sim \delta(A)\,r+(1-\delta(A))w$; equivalently, $\mu_B(B^c)u(q)=\delta(B)\,\mu(B^c)u(q)=\delta(A)u(r)=\delta(A)\mu(B^c)u(q)$. Hence, $\delta(A)=\delta(B)=\delta$. Finally, we set $\delta(S)=\delta$ and obtain a Weighted  \nameref{IEU} representation.

$(i)\Rightarrow(ii)$. SEU postulates are trivially satisfied. Since $\mu$ has full-support, $\mu_A\neq \mu$ for any $A\subset S$. We now shall prove the necessity of the other three axioms. By the argument above, $E$ is $\succsim_A$-unfavored iff $\mu^*_A(E)=0$ where $\mu^*_A=\arg\min_{\pi\in \Delta(A)} d_\mu(\pi)$. Equivalently, $E$ is $\succsim_A$-favored iff $\mu^*_A(E)=1$. 

Partial Consequentialism is satisfied because $A$ is $\succsim_A$-favored; i.e., $\mu^*_A\in\Delta(A)$. 

To prove Partial Dynamic Coherence, take any $A_1, \ldots, A_n\subseteq S$ such that $A_i$ is $\succsim_{A_{i+1}}$-favored for each $i\le n-1$ and $A_n$ is $\succsim_{A_1}$-favored. In other words, $\mu^*_{A_{i+1}}(A_i)=1$ for each $i\le n-1$ and $\mu^*_{A_1}(A_n)=1$. Note that $\mu^*_{A_{i+1}}(A_i)=1$ means that $\mu^*_{A_{i+1}}\in \Delta(A_{i})$. Since $\mu^*_{A_i}$ is the unique minimizer of $d_\mu$ in $\Delta(A_{i})$, we have $d_\mu(\mu^*_{A_i})\le d_\mu(\mu^*_{A_{i+1}})$, the inequality is strict when $\mu^*_{A_i}\neq \mu^*_{A_{i+1}}$. We will obtain a contradiction if there is at least one strict inequality. Hence, $\mu^*_{A_1}=\ldots=\mu^*_{A_n}$, which implies $\succsim_{A_1}=\succsim_{A_{n}}$.

To prove Relative Tradeoff Consistency, take any $A, B\in\Sigma$, $p, q\in\Delta$, and $\alpha\in (0, 1)$ such that 
\[w\,A\,q\sim p\text{ and } w A\,q\sim_A \alpha\,p+(1-\alpha)w;\]
equivalently, 
$\mu(A^c)u(q)=u(p)$ and $\mu_A(A^c)u(q)=\alpha u(p)$. Since $\mu_A(A^c)=\delta\,\mu(A)$, we have $\alpha=\delta$.
Take any $r$ such that $w\,B\,q\sim r$; equivalently, $\mu(B^c)u(q)=u(r)$. Since $\mu_B(B^c)=\alpha\,\mu(B)$, we have $\mu_B(B^c)u(q)=\alpha u(r)$; equivalently, $w\,B\,q\sim_B \alpha\,r+(1-\alpha)w$.

\bibliographystyle{ecta}
\bibliography{econref}

\pagebreak
\section{Bayesian Persuasion: A Richer Message Space - Online Publication Only}\label{rich}

When $f$ is not linear, the revelation principle may be violated (see \cite{de2022non}). Hence, the assumption $|M|=2$ is not without loss of generality. We show that the conclusions of the previous section do not change substantively when $|M|\ge 3$.

Suppose  $|M|\ge 3$ and $a^*_{\pi}(m_s)=a$ holds for at most $k\in [2, |M|-1]$ distinct messages $m_s$. We assume $f$ is continuous. The sender's problem reduces to
\[\max_{\pi} \sum^k_{s=1}\big(\sum^n_{i=1} \rho_i\, \pi_{m_s}(\omega_i)\big)\]
\[\text{subject to }\sum_{i\in A}\delta_i\,f\big(\pi_{m_s}(\omega_i)\big)\ge \sum^n_{i\in A^c}\delta_i\, f\big(\pi_{m_s}(\omega_i)\big)\,\text{ for each }s\le k.\]

We show that the optimal signal structures in this case are similar to ones we obtained in Propositions 7 and 8. 

\begin{prps}\label{concave-rich} Suppose $f$ is strictly concave. For any optimal signal structure $\pi^*$, there is $\bar{\omega}\in A^c$ such that 
\[\pi^*_{m_s}(\omega_1)=\frac{1}{k}\text{ for any }\omega_1\in A\text{ and }s\le k\text{, and }\]
\[\pi^*_{m_s}(\omega_2)\in \{0, 1\}\text{ for any }\omega_2\in A^c\setminus\{\bar{\omega}\}\text{ and }s\le k.\]   
\end{prps}

\autoref{concave-rich} shows that, when $f$ is strictly concave, the sender randomizes at states in $A$ and never randomizes at states in $A^c\setminus{\bar{\omega}}$. In contrast, when $f$ is strictly convex, the sender never randomizes at states in $A$, but instead randomizes at states in $A^c$.  This is shown in \autoref{convex-rich} below.

\begin{prps}\label{convex-rich} Suppose $f$ is strictly convex and $f'(0)=0$. For any optimal signal structure $\pi^*$, 
\[\pi^*_{m_s}(\omega_1)\in \{0, 1\}\text{ for any }\omega_1\in A\text{ and }s\le k\]
\[\pi^*_{m_1}(\omega_2)=\pi^*_{m_s}(\omega_2)\in (0, 1)\text{ for any }\omega_2\in A^c\text{ and }s\le k.\]   
\end{prps} 

The intuition behind the above results is the same as the intuition behind Propositions 7 and 8 since strictly concave (convex) $f$ leads to a strictly convex (concave) objective function. The following example further illustrates the difference between the case $\beta>1$ and the case $\beta<1$.

\begin{exm}[continues=persuasion]  Suppose now $|M|=3$ and $k=2$. When $\beta<1$, 
\[\pi^H_{m_1}=\pi^H_{m_2}=\frac{1}{2},\,\pi^M_{m_1}\in (0, 1), \text{ and }\pi^M_{m_2}=\pi^L_{m_1}=\pi^L_{m_2}=0.\]  
However, when $\beta>1$, 
\[\pi^H_{m_1}=1\text{ and }\pi^H_{m_2}=0\text{ and }\pi^M_{m_1}=\pi^M_{m_2}\in (0, 1) \text{ and }\pi^L_{m_1}=\pi^L_{m_2}\in (0, 1).\] 
\end{exm}

\subsection{Proof of \autoref{concave-rich}}

We first solve
\[\max_{\pi} \sum_{s\le k}\sum_{i\in A}\delta_i\,f\big(\pi_{m_s}(\omega_i)\big)=\max_{\pi} \sum_{i\in A }\delta_i \sum^k_{s=1} f\big(\pi_{m_s}(\omega_i)\big).\]
Since $f$ is strictly concave, $\pi^{*}_{m_s}(\omega_i)=\frac{1}{k}$ for any $s\le k$ and $i\in A$. Hence, 
\[\max_{\pi} \sum_{s\le k}\sum_{i\in A}\delta_i\,f\big(\pi_{m_s}(\omega_i)\big)=k\,f(\frac{1}{k})\,\sum_{i\in A}\delta_i=M.\]
Then we shall solve
\[\max_{\pi} \sum^k_{s=1}\big(\sum_{i\in B} \rho_i\, \pi_{m_s}(\omega_i)\big)=\sum_{i\in B}\rho_i \big(\sum^k_{s=1}\pi_{m_s}(\omega_i)\big)\]
\[\text{subject to }\sum_{s\le k}\sum_{i\in B}\delta_i\, f\big(\pi_{m_s}(\omega_i)\big)=\sum_{i\in B}\delta_i\, \big(\sum_{s\le k}f\big(\pi_{m_s}(\omega_i)\big)\big)\le M.\]
The solution to the above problem will be the solution to the problem below for some $M_i$:
\[\max \sum^k_{s=1}\pi_{m_s}(\omega_i)\text{ subject to }\sum_{s\le k}f\big(\pi_{m_s}(\omega_i)\big)\le M_i.\]

Since $f$ is strictly concave, there is some $s$ such that $\pi_{m_s}(\omega_i)=\min\{1, f^{-1}(M_i)\}$ and $\pi_{m_{s'}}(\omega_i)=0$ for each $s'\neq s$. By \autoref{convex}, there is $\bar{\omega}\in A^c$ such that $M_i$ is either $f(1)$ or $f(0)$ for each $i\in A^c\setminus\{\bar{\omega}\}$.

\subsection{Proof of \autoref{convex-rich}}

We first solve
\[\max_{\pi} \sum_{s\le k}\sum_{i\in A}\delta_i\,f\big(\pi_{m_s}(\omega_i)\big)=\sum_{i\in A }\delta_i \sum^k_{s=1} f\big(\pi_{m_s}(\omega_i)\big).\]
Since $f$ is strictly convex, there is $s\le k$ such that $\pi^{*}_{m_s}(\omega_i)=1$ and $\pi^{*}_{m_{s'}}(\omega_i)=0$ for each $s'\neq s$. Hence, 
\[\max_{\pi} \sum_{s\le k}\sum_{i\in A}\delta_i\,f\big(\pi_{m_s}(\omega_i)\big)=f(1)\,\sum_{i\in A}\delta_i=M.\]
Then we shall solve
\[\max_{\pi} \sum^k_{s=1}\big(\sum_{i\in B} \rho_i\, \pi_{m_s}(\omega_i)\big)=\sum_{i\in B}\rho_i \big(\sum^k_{s=1}\pi_{m_s}(\omega_i)\big)\]
\[\text{subject to }\sum_{s\le k}\sum_{i\in B}\delta_i\, f\big(\pi_{m_s}(\omega_i)\big)=\sum_{i\in B}\delta_i\, \big(\sum_{s\le k}f\big(\pi_{m_s}(\omega_i)\big)\big)\le M.\]
The solution to the above problem will be the solution to the problem below for some $M_i$:
\[\max \sum^k_{s=1}\pi_{m_s}(\omega_i)\text{ subject to }\sum_{s\le k}f\big(\pi_{m_s}(\omega_i)\big)\le M_i.\]

Since $f$ is strictly convex, $\pi_{m_{s}}(\omega_i)=\min\{\frac{1}{k}, f^{-1}(\frac{M_i}{k})\}$ for each $s\le k$. By \autoref{concave}, $M_i>f(0)$.

\end{document}